\documentclass[journal, onecolumn]{IEEEtran}

\usepackage{amsthm}

\usepackage{tikz}
\usetikzlibrary{arrows.meta, positioning, shapes,snakes, calc}
\usepackage{subfig}
\usepackage{amssymb}
\usepackage{mathtools}
\usepackage{cases}
\usepackage{autobreak}
\usepackage{mathrsfs}  
\usepackage{booktabs}
\usepackage{multirow}
\usepackage{color}
\usepackage{hyperref}
\usepackage{cite}
\usepackage{amsmath}
\usepackage{extarrows}
\usepackage{bm}
\usepackage{comment}
\usepackage{algorithmic}
\usepackage{algorithm}
\usepackage{caption}
\usepackage{graphicx}
\usepackage{pifont}
\usepackage{bbm}

\newcommand{\BA}{{\bold{A}}}
\newcommand{\BB}{{\bold{B}}}
\newcommand{\BC}{{\bold{C}}}
\newcommand{\BD}{{\bold{D}}}
\newcommand{\BE}{{\bold{E}}}
\newcommand{\BF}{{\bold{F}}}
\newcommand{\BG}{{\bold{G}}}
\newcommand{\BH}{{\bold{H}}}
\newcommand{\BI}{{\bold{I}}}

\newcommand{\BL}{{\bold{L}}}
\newcommand{\BM}{{\bold{M}}}

\newcommand{\BP}{{\bold{P}}}
\newcommand{\BQ}{{\bold{Q}}}
\newcommand{\BR}{{\bold{R}}}
\newcommand{\BS}{{\bold{S}}}
\newcommand{\BT}{{\bold{T}}}
\newcommand{\BU}{{\bold{U}}}
\newcommand{\BV}{{\bold{V}}}

\newcommand{\BX}{{\bold{X}}}
\newcommand{\BY}{{\bold{Y}}}
\newcommand{\BZ}{{\bold{Z}}}

\newcommand{\Be}{{\bold{e}}}

\newcommand{\Bh}{{\bold{h}}}

\newcommand{\Bn}{{\bold{n}}}

\newcommand{\Bu}{{\bold{u}}}
\newcommand{\Bv}{{\bold{v}}}

\newcommand{\Bx}{{\bold{x}}}
\newcommand{\By}{{\bold{y}}}
\newcommand{\Bz}{{\bold{z}}}

\newcommand{\BCD}{\boldsymbol{\mathcal{D}}}
\newcommand{\BCE}{\boldsymbol{\mathcal{E}}}

\newcommand{\BCI}{\boldsymbol{\mathcal{I}}}

\newcommand{\BSigma}{\bold{\Sigma}}

\newcommand{\BXi}{\bold{\Xi}}

\newcommand{\BDelta}{\boldsymbol{\Delta}}
\newcommand{\BPhi}{\boldsymbol{\Phi}}

\newcommand{\BPsi}{\boldsymbol{\Psi}}

\newcommand{\Bvarsigma}{\boldsymbol{\varsigma}}
\newcommand{\Balpha}{\boldsymbol{\alpha}}

\newcommand{\Bbeta}{\boldsymbol{\beta}}

\newcommand{\Bmu}{\boldsymbol{\mu}}
\newcommand{\Bnu}{\boldsymbol{\nu}}

\newcommand{\cB}{\mathcal{B}}
\newcommand{\cC}{\mathcal{C}}
\newcommand{\cD}{\mathcal{D}}
\newcommand{\cE}{\mathcal{E}}
\newcommand{\cF}{\mathcal{F}}

\newcommand{\cH}{\mathcal{H}}
\newcommand{\cI}{\mathcal{I}}

\newcommand{\cK}{\mathcal{K}}
\newcommand{\cL}{\mathcal{L}}

\newcommand{\cN}{\mathcal{N}}
\newcommand{\cO}{\mathcal{O}}
\newcommand{\cP}{\mathcal{P}}

\newcommand{\cS}{\mathcal{S}}
\newcommand{\cT}{\mathcal{T}}

\newcommand{\cW}{\mathcal{W}}

\newcommand*{\rom}[1]{\expandafter\@slowromancap\romannumeral #1@}

\newcommand{\E}{\mathbb{E}}
\newcommand{\indvj}{\boldsymbol{\mathsf{j}}}
\newcommand{\R}{\mathbb{R}}
\newcommand{\C}{\mathbb{C}}
\newcommand{\U}{\mathbb{U}}
\newcommand{\PP}{\mathbb{P}}
\newcommand\innerprod[2]{\left\langle #1, #2 \right\rangle}
\newcommand{\Vecm}{\mathrm{Vec}}
\newcommand{\norm}[1]{\left\lVert#1\right\rVert}
\newcommand{\abs}[1]{\left|#1\right|}
\newcommand{\lnorm}[1]{\lVert#1\rVert}
\newcommand{\dd}{\mathrm{d}}
\newcommand{\Ind}{\mathbbm{1}}

\newcommand{\dz}{\mathsf{d}_z}

\DeclareMathOperator{\Tr}{Tr}
\DeclareMathOperator{\dist}{dist}
\DeclareMathOperator{\Span}{Span}

\DeclareMathOperator{\rank}{rank}
\DeclareMathOperator{\diag}{diag}
\DeclareMathOperator{\Var}{Var}

\newtheorem{remark}{Remark}
\newtheorem{theorem}{Theorem}
\newtheorem{lemma}{Lemma}
\newtheorem{proposition}{Proposition}
\newtheorem{assumption}{Assumption}
\newtheorem{corollary}{Corollary}

\renewcommand{\leq}{\leqslant}
\renewcommand{\geq}{\geqslant}
\newcommand\numberthis{\addtocounter{equation}{1}\tag{\theequation}}
\allowdisplaybreaks[4]

\author{\IEEEauthorblockN {Zeyan Zhuang,  Anzheng Tang, Xin Zhang, Dongfang Xu, and Shenghui Song, \textit{Senior Member}, \textit{IEEE}}

\thanks{
Zeyan Zhuang, Anzheng Tang, and Shenghui Song are with the Department of Electronic and Computer Engineering,
The Hong Kong University of Science and Technology, Hong Kong. Shenghui Song is also with the Division of Integrative Systems and Design, The Hong Kong University of Science and Technology, Hong Kong (e-mail: zzhuangac@connect.ust.hk; eeaztang@ust.hk; eeshsong@ust.hk).
\par
Xin Zhang is with the School of Cyber
Science and Technology, Beihang University, Beijing 100191, China (e-mail: zhangxinn@buaa.edu.cn).
\par
Dongfang Xu is with the College of Artificial Intelligence, Nanjing University of Aeronautics and Astronautics, Nanjing 210000, China (e-mail: xudongfang@nuaa.edu.cn).
}
}

\begin{document}
\title{A Unified Framework for Characterizing General MIMO Channels}

\maketitle

\begin{abstract}
Modern multiple-input multiple-output (MIMO) systems are evolving toward higher-dimensional and more flexible architectures, such as distributed and holographic MIMO, offering substantial benefits while giving rise to increasingly complex channel correlation structures. This growing architectural diversity makes it difficult to characterize the fundamental limits of different MIMO systems on a case-by-case basis, motivating a unified analytical framework applicable across a broad range of architectures. This paper addresses this research gap by investigating a generally correlated Rayleigh channel whose vectorized form follows a general Gaussian distribution with an arbitrary covariance matrix. For that purpose, we first prove that the spectral norm of the channel matrix is bounded in the asymptotic regime where the numbers of transmit and receive antennas grow proportionally. We then derive a deterministic approximation for the ergodic mutual information, with an explicit convergence rate governed by the structure of the channel correlation. The approximation is characterized by a pair of matrix-valued self-consistent equations, for which we establish the existence and uniqueness of the solution and propose an iterative numerical algorithm. Furthermore, we develop a framework based on positive linear maps to analyze the stability of these equations, which can facilitate the spectral analysis of broader classes of random matrices. The proposed characterization unifies many existing results for structured MIMO channels as special cases while remaining applicable to a broad class of general MIMO architectures that are difficult to evaluate using existing analytical frameworks. To demonstrate its utility, we apply the developed theory to two representative systems, namely downlink distributed MIMO and uplink heterogeneous MIMO. Numerical results confirm the accuracy of the derived deterministic approximations.
\end{abstract}

\begin{IEEEkeywords}
Random Matrix Theory, General Gaussian matrices, Multiple-Input Multiple-Output, Mutual Information.
\end{IEEEkeywords}

\section{Introduction}

Multiple-input multiple-output (MIMO) technology has become a key enabler of wireless communication systems owing to its ability to improve spectral and energy efficiency, link reliability, and spatial multiplexing gains \cite{paulMIMO2004}. In particular, by employing multiple antennas at both the transmitter and the receiver, MIMO systems provide additional spatial degrees of freedom (DoF) to enable higher diversity and multiplexing gains \cite{JafardofMIMO}. Driven by the rapidly growing demands of emerging wireless applications, MIMO systems are evolving toward higher dimensions and more flexible hardware configurations \cite{WangMIMO2024}. For instance, massive MIMO (mMIMO) supports high data rate communications by equipping the base station (BS) with hundreds of antennas \cite{marze2010twcmMIMO}. Compared with conventional MIMO, mMIMO provides substantial beamforming and multiplexing gains while benefiting from favorable propagation and channel hardening \cite{Hoydis2013MIMO}. On the other hand, to mitigate intercell interference and provide more uniform service quality, cell-free and distributed mMIMO architectures have been developed, in which geographically separated antenna sets cooperatively serve users \cite{zhang2013jsacMIMO}. Meanwhile, decentralized baseband processing  partitions a large antenna array into multiple clusters equipped with local processing units, thereby reducing interconnection overhead and the computational burden of high-dimensional signal processing \cite{Zhuang2025Decentralized,Zhao2025}. Furthermore, with the advances in reconfigurable antenna hardware and electromagnetic metasurfaces, intelligent multi-antenna technologies have been widely investigated. Representative examples include reconfigurable intelligent surfaces (RISs) \cite{Basar2019IRS}, holographic MIMO (HMIMO) systems \cite{Pizzo2022Holo, Zhang2025HolographicMIMO}, and flexible antenna systems (FASs) \cite{DingFAS2025}. These emerging architectures promise lower costs and considerable performance gains, making them important candidates for future wireless systems.

\par

Alongside the benefits of modern MIMO systems, their increasing dimensionality and flexible hardware architectures give rise to more complex channel correlation structures. Consequently, it becomes increasingly difficult to characterize the fundamental limits of different types of MIMO systems on a case-by-case basis and a unified framework suitable for a large range of architectures will be very beneficial. Unfortunately, existing analyses commonly focus on specific channel fading structures, such as separable correlations \cite{Hachem2008ANewApproach}, variance profiles \cite{hachem2008varianceProfile}, or block-structured correlations \cite{zhang2013jsacMIMO}, and therefore cannot be directly applied to tackle general MIMO channels. In this work, we address this limitation by studying a generally correlated Rayleigh MIMO channel that allows arbitrary correlations among the channel coefficients, which encompasses a broad class of MIMO systems. In particular, we derive deterministic approximations for the ergodic mutual information (EMI) of the general MIMO systems with an explicit convergence rate.

\subsection{Related Work}

\textit{Performance Analysis of MIMO Channels:}
The EMI of different types of MIMO channels has been extensively investigated, yielding innovative results. Early studies primarily considered independent and identically distributed (i.i.d.) Rayleigh channels \cite{foschini1998limits}. In the large-system regime, the asymptotic MI was characterized through the Mar\v{c}enko-Pastur law  \cite{Tselinear1999, marchenko1967distribution}. To incorporate spatial correlations at both the transmitter and the receiver, the Kronecker-correlated Rayleigh channel was investigated in \cite{Moustakas2003titcorrelated}. In particular, Moustakas \textit{et al.} derived a deterministic approximation for the EMI and established a central limit theorem (CLT) for the MI utilizing the replica method. These results were then rigorously established in \cite{Hachem2008ANewApproach} using Gaussian tools. Similarly, the statistical properties of the MI for Kronecker Rician MIMO channels were analyzed in \cite{moustakas2005random,Dumont2010rician}, while Rician channels with a variance profile were studied in \cite{hachem2007deterministic,hachem2008varianceProfile}. Due to multipath propagation, MIMO channels generally exhibit frequency-selective fading.  The works in \cite{moustaka2007freqselec, Dupuy2011freqselec} established deterministic approximations for the MI of correlated frequency-selective MIMO channels. In \cite{Mckay2008variance}, McKay \textit{et al.} employed finite-dimensional RMT to characterize the mean and variance of the MI of MIMO orthogonal frequency-division multiplexing (OFDM) systems while accounting for statistical dependence among channel matrices across different subcarriers.

The aforementioned studies primarily focused on point-to-point MIMO systems. The RMT-based performance analyses were subsequently extended to multi-user and distributed MIMO systems. In \cite{Couillet2011DE}, deterministic approximations for the EMI of MIMO multiple-access channels (MACs) were derived under the Kronecker model. These results were then used to characterize the ergodic rate region. In \cite{zhang2013jsacMIMO}, Zhang \textit{et al.} further studied a distributed MIMO MAC in which the base station (BS) consists of several geographically distributed antenna sets. In particular, a deterministic approximation for the ergodic sum rate was derived under a spatially correlated Rician channel model. In contrast to distributed MIMO models that assume independent channels across different antenna sets, decentralized MIMO systems may involve spatial correlations among antenna clusters. The performance of such systems was investigated in \cite{Zhuang2025Decentralized}.

\par

More recently, emerging MIMO architectures such as HMIMO have been proposed \cite{Pizzo2022Holo}. In HMIMO systems, antenna elements are densely deployed over a compact physical aperture, resulting in nonseparable spatial correlations and pronounced electromagnetic coupling \cite{pizzojsaitholoMIMOmc}. Closed-form approximations for the EMI and outage probability of HMIMO systems were derived in \cite{Zhang2025HolographicMIMO} under a noncentered and nonseparable channel model. The effect of mutual coupling (MC) on system performance was further investigated in \cite{Zhang2025HolographicMC}. 
These developments highlight the growing diversity of MIMO channel correlation structures. Analyzing each architecture through a separate model is increasingly cumbersome, calling for a unified analytical framework. Specifically, let $\BH$ denote the channel matrix and assume that $\Vecm(\BH)$ follows a zero-mean complex Gaussian distribution with covariance $\BC$. Then, all correlation structures discussed above can be represented through an appropriate choice of $\BC$. Moreover, the increasing dimensions of these systems make RMT a natural tool for the performance analysis. Accordingly, this paper derives a deterministic approximation for the EMI under this general model with an explicit convergence rate.

\textit{Spectral Analysis of Gaussian Matrices:}
The above MIMO performance analyses generally rely on characterizing the spectrum of Gaussian channel matrices, representing a canonical class of random matrices. Specifically, many spectral convergence results were first established under Gaussian assumptions and subsequently generalized to broader classes of non-Gaussian ensembles \cite{bai2010spectral, Wen2013NonGaussian}, a phenomenon commonly referred to as \textit{universality} in RMT. For example, the almost sure convergence of the smallest eigenvalue of a Wishart matrix was first established in \cite{JWsilverstein1985wishartsmall}. After appropriate scaling, the largest eigenvalues of the Gaussian orthogonal ensemble and Gaussian sample covariance matrices were shown in \cite{tracy1996orthogonal} and \cite{johansson2000shape}, respectively, to converge weakly to the celebrated Tracy-Widom laws.

\par
Gaussian tools provide a powerful framework for analyzing Gaussian matrices, relying primarily on two ingredients. The \textit{integration by parts formula} enables the derivation of self-consistent equations for the resolvent, whereas the \textit{Poincar\'e--Nash inequality} controls the fluctuations of the approximation errors. These techniques were systematically developed in the early work by Pastur \textit{et al.} \cite{pastur2011eigenvalue} and have since become standard tools in RMT. In wireless communications, Gaussian tools have been used to characterize the EMI and outage performance of two-hop MIMO systems \cite{Xin2023Double-Scattering,Zhuang2025Fundamental}. Furthermore, they have  been applied in multivariate statistics and machine learning, including subspace estimation \cite{Vallet2012Subspace}, the analysis of kernel methods \cite{couillet2021twoway}, and the spectral analysis of tensor data \cite{JMLRHugotensor}. In this work, we employ Gaussian tools to analyze the generally correlated MIMO channel. However, owing to the entrywise correlations in $\BH$, the resulting error terms involve nondiagonal matrix-valued maps and therefore require a more refined analysis.

\textit{Stability Analysis of the Approximation Equations:}
The application of Gaussian tools typically yields self-consistent equations whose solutions characterize the resolvent approximation \cite{hachem2007deterministic}. A central step in establishing the convergence rate of the approximation is to analyze the stability of the equations, namely, to show that a small perturbation of the self-consistent system induces a controlled perturbation of its solution. For vector-valued self-consistent equations, a framework based on positive matrices was developed in \cite{hachem2008varianceProfile}. An important advantage of this framework is that it provides quantitative stability bounds with polynomial dependence on the spectral parameter, thereby facilitating the analysis of edge and outlier eigenvalues \cite{AblaNoEigen,Vallet2012Subspace,chapon2012outliers,hachem2013bilinear}. However, under the general covariance matrix $\BC$ considered in this work,  the resulting self-consistent equations are matrix-valued, and conventional  arguments are therefore no longer applicable. To overcome this difficulty, we develop a stability framework based on positive linear maps. Stability analysis of such equations is also an important ingredient in modern RMT, including the characterization of local spectral behavior \cite{benaych2016lectures,erdHos2017dynamical}.

\subsection{Contributions}
The main contributions of this paper are summarized as follows:
\begin{itemize}
    \item[1)] We establish a probabilistic norm bound for a generally correlated Gaussian matrix. Specifically, under proportional dimensional growth and a uniformly bounded covariance matrix $\BC$, we prove that its spectral norm is almost surely bounded and has uniformly bounded moments of every finite order. This result complements the classical Gaussian tools.

    \item[2)] Building on this norm control, we derive a deterministic approximation for the trace form of the resolvent of $\BH\BH^H$, together with an explicit convergence rate governed by the structural parameter $\mathsf{E}_N$ defined in \eqref{eq: convergence E_N}, which captures the intrinsic structure of covariance $\BC$. The resolvent approximation further yields deterministic approximations for linear spectral statistics (LSS) and the EMI, recovering existing results for separable \cite{Hachem2008ANewApproach}, column-independent \cite{AblaNoEigen}, mixed-separable \cite{zhang2013jsacMIMO}, and additive-separable \cite{Dupuy2011freqselec} channel models as special cases.

    \item[3)] The resulting resolvent approximation is characterized by a pair of matrix-valued self-consistent equations. We establish the existence and uniqueness of their solution  within the class of matrix-valued Stieltjes transforms and provide an iterative algorithm for its numerical computation. Moreover, we develop a framework based on positive linear maps to control the stability of these equations. This stability framework is central to deriving the explicit convergence rate.

    \item[4)] Finally, we use the resolvent approximation to derive a deterministic approximation for the EMI and apply it to a downlink distributed MIMO system and an uplink heterogeneous MIMO system. For both settings, we obtain an $\mathcal{O}(N^{-1/2})$ convergence rate, where $N$ is the number of receive antennas.  Simulation results validate the accuracy of the approximation. Furthermore, the analysis of general Gaussian channels can be specialized to a broad range of MIMO systems beyond those considered here.
    
\end{itemize}

\textit{Paper Outline:}
The paper is organized as follows. In Section \ref{Sec: Sys_Model}, we introduce the general MIMO channel model and assumptions used throughout the paper. In Section \ref{Sec:The res}, we
present the main theoretical results, including the boundedness of the channel spectral norm and the deterministic approximation of the resolvent and the MI. In Section \ref{Sec: App comm}, we apply the results to the analysis of downlink distributed MIMO systems and uplink heterogeneous MIMO communications. Numerical simulations are provided in
Section \ref{Sec: simulations} to validate the accuracy of the theoretical results and Section \ref{Sec: Conclu} concludes this paper. 
\par
\textit{Notation:}
Throughout this paper, lowercase and uppercase boldface letters denote vectors and matrices, respectively. The set $\mathbb{U}(n)$ denotes the group of $n$-by-$n$ unitary matrices. The superscripts $(\cdot)^\top$, $(\cdot)^*$, and $(\cdot)^H$ stand for transpose, entrywise complex conjugation, and Hermitian transpose, respectively. For a matrix \(\mathbf A\), \(\operatorname{rank}(\mathbf A)\), \(\|\mathbf A\|\), and \(\|\mathbf A\|_F\) denote its rank, spectral norm, and Frobenius norm, respectively. The trace of a square matrix \(\mathbf A\) is denoted by \(\operatorname{Tr}\mathbf A\). For a Hermitian matrix \(\mathbf A\), \(\lambda_{\min}(\mathbf A)\) and \(\lambda_{\max}(\mathbf A)\) denote its smallest eigenvalue and largest eigenvalue, respectively. The $N$-by-$N$ identity matrix is denoted by $\BI_N$, and the subscript is omitted when the dimension is clear from the context. For a vector $\Bx$, $\norm{\Bx}$ denotes the Euclidean norm. For a matrix $\BH$, $\Vecm(\BH)$ stacks the columns of $\BH$ into a vector. The notation $[\BA]_{i,j}$ or $A_{i, j}$ represents the $(i,j)$-th entry of $\BA$, while $\BA_{ij}$ is used for the $(i,j)$-th block of a block matrix when the block partition is specified. The Kronecker product and the Hadamard product are denoted by $\otimes$ and $\odot$, respectively. For two matrices of the same size, $\innerprod{\BA}{\BB}=\Tr(\BA\BB^H)$ denotes the Hilbert-Schmidt inner product. The symbols $\E[\cdot]$, $\Var(\cdot)$, and $\PP(\cdot)$ denote expectation, variance, and probability, respectively. The notation $\Bx\sim\mathcal{CN}(\Bmu,\BSigma)$ means that $\Bx - \Bmu$ is a circularly symmetric complex Gaussian vector with covariance $\BSigma$. For two Hermitian matrices $\BA$ and $\BB$, $\BA\succeq\BB$ means that $\BA-\BB$ is Hermitian non-negative. 
We use $\mathcal{O}(\cdot)$ for the standard big-O notation, with the implicit constants being independent of the system dimensions unless otherwise specified. The asymptotic notation $a \asymp b$ means that there exists a constant $K>0$ such that
$
K^{-1}\leq\liminf|a/b|
\leq\limsup|a/b|\leq K.
$
The indicator function  is denoted by $\Ind_{\{\cdot\}}$. For a positive integer \(n\), we write \([n]=\{1,2,\ldots,n\}\). We denote the upper half-plane by $\C_+ = \{z \in \C : \Im(z) > 0\}$ and the non-negative real axis by $\R_+ = \{x \in \R : x \geq 0\}$.

\section{System Model}
\label{Sec: Sys_Model}
Consider a general MIMO system with $M$ transmit and $N$ receive antennas, where the channel matrix between the transmitter and receiver is denoted by $\BH \in \C^{N \times M}$.  Thus, the received signal is given by
\begin{equation}
    \By = \BH \Bx + \Bn \in \C^{N},
\end{equation}
where $\Bx \in \mathbb{C}^M$ represents the transmit signal with zero mean and covariance $\E[\Bx\Bx^H] = \BI_M$ and $\Bn \sim \mathcal{CN}(0, \sigma^2 \BI_N)$ denotes the additive white Gaussian noise (AWGN) at the receiver.
\par
In this work, we consider a generalized Rayleigh fading model with an arbitrary correlation structure, which includes existing channel models as special cases, e.g., i.i.d. model \cite{foschini1998limits}, Kronecker model \cite{MIMODaShan2000}, and Weichselberger model \cite{Kermoal2002, Weichs2006channel}. Let $\Bh=\Vecm(\BH)$ be the vectorized channel matrix. We assume 
\begin{align}
    \label{eq: R channel model}
    \Bh \sim \mathcal{CN}\left(\mathbf{0}, {\BC}/ M\right),
\end{align}
where $\BC \in \C^{NM \times NM}$ is deterministic and Hermitian non-negative. 
\par
We assume that perfect channel state information (CSI) is available at the receiver, whereas the transmitter has access to statistical CSI. The MI of the considered MIMO channel is given by
\begin{align}
    \mathcal{I}(\sigma^2)
    &=
    \log\det\left(
        \BI_N+\frac{1}{\sigma^2}\BH\BH^H
    \right) =
    \sum_{i=1}^{N}
    \log\left(1+ \lambda_i / \sigma^2 \right) =
    N\int_{\R_+}
    \log\left(1+ \lambda / \sigma^2 \right)
    \widehat{F}^N(\dd\lambda), 
    \label{eq: Def of MI}
\end{align}
where $
    0\leq\lambda_1\leq\lambda_2\leq\cdots\leq\lambda_N $
denote the eigenvalues of $\BH\BH^H$. The corresponding empirical spectral distribution (ESD) is defined as
\begin{align}
    \widehat{F}^N(x) =
    \frac{1}{N}\sum_{i=1}^{N}\Ind_{\{\lambda_i\leq x\}}.
\end{align}
In this work, we will characterize the deterministic limit of the ESD $\widehat{F}^N$ and then derive an explicit approximation for the EMI of the system.  To facilitate the analysis, we make the following assumptions.
\begin{assumption}
\label{assumpt: A1}
The dimensions $N$ and $M$ approach infinity with 
    \begin{equation}
        0< \liminf_{N \to \infty} \frac{N}{M} \leq \limsup_{N \to \infty} \frac{N}{M} < \infty.
    \end{equation}
\end{assumption}

\begin{assumption} 
\label{assumpt: A2}
The spectral norm of $\BC$ satisfies
    \begin{equation}
        \limsup_{N \to \infty } \norm{\BC} < \infty.
    \end{equation}
\end{assumption}
\begin{assumption}
\label{assumpt: A3}
The minimal eigenvalue of $\BC$ satisfies
    \begin{equation}
        \liminf_{N \to \infty } \lambda_{\min}(\BC) > 0.
    \end{equation}
\end{assumption}
Assumption \ref{assumpt: A1} describes the asymptotic regime where both $N$ and $M$ tend to infinity at the same pace, which is a common setting in the study of large-dimensional random matrices. In what follows, we use $N\to\infty$ to denote this asymptotic regime. Strictly speaking, the eigenvalues $\{\lambda_i\}_{1 \leq i\leq N} \equiv \{\lambda_i^{(N)}\}_{1 \leq i\leq N}$ and matrices $\BH \equiv \BH_N$ and $\BC \equiv \BC_N$ depend on $N$. For notational simplicity, we suppress the index $N$ when there is no ambiguity and recover it when necessary. Assumptions \ref{assumpt: A2} and \ref{assumpt: A3} ensure that the eigenvalues of the correlation matrix $\BC$ lie within a compact interval in $(0, \infty)$, thereby avoiding the divergence of eigenvalues and cases of rank deficiency.

\section{Main Results}
\label{Sec:The res}
In this section, we present the main theoretical results on the spectral behavior of the covariance matrix $\BH\BH^H$ and the deterministic approximation for the EMI. In particular, we first establish a probabilistic bound for the spectral norm of $\BH$ in Section \ref{Sec_A_bound_for_the_Spectral_Norm}. We then characterize the asymptotic behavior of $\BH\BH^H$ through its resolvent in Section \ref{Sec_Converg_of_Res}, which further yields the  deterministic approximation for the EMI in Section \ref{Sec_Conv_EMI}.

\subsection{A Bound for the Spectral Norm}
\label{Sec_A_bound_for_the_Spectral_Norm}
To study the eigenvalue distribution of the covariance matrix $\BH\BH^H$, it is necessary to first establish a probabilistic bound for the spectral norm of $\BH$. The following theorem shows that the spectral norm of the generally correlated Gaussian matrix is bounded.
\begin{theorem}
\label{thm: bound norm}
    Assume that Assumptions \ref{assumpt: A1} and \ref{assumpt: A2} hold. Then, there exists a constant $K > 0$ (independent of $N, M$) such that
    \begin{align*}
        \PP \left( \limsup_{N \to \infty} \norm{\BH} \leq K \right) = 1. \numberthis \label{eq: limsup prob bound H}
    \end{align*}
    Additionally, for every positive number $p$, the following moment bound holds
    \begin{equation}
        \sup_{N \geq 1} \E  \left[ \norm{\BH}^p  \right] \leq K_p,
    \end{equation}
    where the constant $K_p > 0$ is also  independent of $N$ and $M$.
\end{theorem}
\textit{Proof:} The proof of Theorem \ref{thm: bound norm} is given in Appendix \ref{app: proof of thm bound norm}.\qed
\par
\begin{remark}
    In \cite[Theorem 2]{RubioCLT2012}, the bound in \eqref{eq: limsup prob bound H} was established for random matrices with i.i.d. entries distributed as $\mathcal{CN}(0,1/M)$. Theorem \ref{thm: bound norm} extends this result to a generally correlated Gaussian matrix, where $\BH$ is obtained by stacking an arbitrary Gaussian vector, allowing arbitrary correlations among the entries of $\BH$.
\end{remark}

\subsection{Deterministic Approximation of the Resolvent}
\label{Sec_Converg_of_Res}

Having established a bound on the spectral norm $\norm{\BH}$, we now turn to the collective behavior of the eigenvalues for $\BH\BH^H$, as characterized by its ESD. A fundamental tool for this purpose is the \textit{resolvent} \cite{bai2010spectral, couillet2011random}. In this subsection, we derive a deterministic approximation for the trace form of the resolvent. Specifically, the resolvent associated with $\BH\BH^H$ is defined as
\begin{align}
    \BQ(z)  =\left( \BH \BH^H - z\BI_N\right)^{-1}, ~~ z \in \C - \R_+.
\end{align}
The normalized trace of the resolvent is precisely the Stieltjes transform of the ESD of $\BH\BH^H$ \cite{bai2010spectral}. In particular, for $z \in \C - \R_+$,
\begin{align}
    \frac{1}{N}\Tr\BQ(z)
    &=
    \frac{1}{N}\sum_{i=1}^{N}\frac{1}{\lambda_i-z}
    =
    \int_{\R_+}\frac{\widehat{F}^{N}(\dd\lambda)}{\lambda-z} = \widehat{m}_N(z).
\end{align}
Consequently, the asymptotic behavior of the ESD can be characterized through the corresponding resolvent \cite{hachem2007deterministic}. In this section, we study the convergence of the resolvent as $N \to \infty$. To facilitate the presentation, we first introduce some notation. Define the length-$N$ index vectors $\Bvarsigma_j = [(j-1)N + 1, (j-1)N + 2, \ldots, jN]^\top$ for $j \in [M]$. For an $MN \times MN$ matrix $\BM$, we denote by $\BM_{ij}$ the $N \times N$ subblock $\BM[\Bvarsigma_i, \Bvarsigma_j]$ for brevity. 
Define two linear maps $\cL: \C^{N \times N} \mapsto \C^{M \times M}$ and $\widetilde{\cL}: \C^{M \times M} \mapsto \C^{N \times N}$ by
\begin{align}
    \label{eq: def of two operators}
    \cL(\BX) = \E_{\bar{\BH}} \left[ \bar{\BH}^H \BX \bar{\BH}\right], ~~ \widetilde{\cL}(\BY) = \E_{\bar{\BH}} \left[ \bar{\BH} \BY \bar{\BH}^H \right],
\end{align}
where $\bar{\BH}$ is an independent copy of ${\BH}$, while $\BX \in \C^{N \times N}$ and $\BY \in \C^{M \times M}$ are arbitrary deterministic matrices. We note that both $\cL$ and $\widetilde{\cL}$ are entirely determined by the correlation matrix $\BC$. In particular, we have 
$
[\cL(\BX)]_{i, j} = \frac{1}{M}\Tr \BC_{ji}\BX $ for each $i, j \in [M]$ and $ \widetilde{\cL}(\BY) = \frac{1}{M} \sum_{a, b = 1}^M  \BC_{ab}[\BY]_{a, b}
$. The following matrix-valued self-consistent equations describe the key quantities that approximate the resolvent
\begin{align*}
\left\{
\begin{array}{rcl}
\BT(z) &=& -\left[z \Big(  \BI_N + \widetilde{\mathcal{L}}[\widetilde{\BT}(z)] \Big)\right]^{-1}, \\
\widetilde{\BT}(z) &=& -  \left[ z\Big(\BI_M + \mathcal{L}[{\BT}(z)] \Big) \right]^{-1},
\end{array}
\right. \numberthis \label{eq: DE systems}
\end{align*}
for $z \in \C - \R_+$.
Let \(\mathcal S_n(\mathbb R_+)\) denote the class of Stieltjes transforms of finite positive semidefinite $n$-by-$n$ matrix-valued measures on \(\mathbb R_+\) (see details in Appendix \ref{app: proof of prop exist unique DE}). Proposition \ref{prop: exist unique DE} establishes the existence and uniqueness of the solution within the given class.
\begin{proposition}
\label{prop: exist unique DE}
    The matrix-valued system of equations \eqref{eq: DE systems} admits a unique solution $(\BT, \widetilde{\BT})$ such that $\mathbf{T} \in \mathcal{S}_N(\R_+)$ and $\widetilde{\BT} \in \mathcal{S}_M(\R_+)$ with total masses $\bm{\mu}_{\BT}(\R_+) = \BI_N$ and $\bm{\mu}_{\widetilde{\BT}}(\R_+) = \BI_M$.
\end{proposition}
\textit{Proof:} The proof of Proposition \ref{prop: exist unique DE} is given in Appendix \ref{app: proof of prop exist unique DE}. \qed
\par
From the proof of Proposition \ref{prop: exist unique DE}, an iterative algorithm for solving \eqref{eq: DE systems} can be derived, which is presented in Algorithm \ref{alg: solve for de}. When there is no confusion, we denote by $(\BT(z), \widetilde{\BT}(z))$ the unique solution from Proposition \ref{prop: exist unique DE} in the following.
\begin{remark}[Computational complexity]
In general, since each iteration of Algorithm~\ref{alg: solve for de} has a complexity of order $\mathcal{O}(N^4)$, the overall complexity is given by $
\mathcal{O}(N^4\log(\epsilon^{-1})),
$
where $\epsilon > 0$ denotes the prescribed solution accuracy. For channel models with additional correlation structures, this complexity can be substantially reduced. In particular, when the columns of $\BH$ are mutually independent \cite{AblaNoEigen} or when $\BH$ follows a Kronecker model \cite{Hachem2008ANewApproach}, the overall complexity is reduced to
$
\mathcal{O}(N^3\log(\epsilon^{-1})).
$
\end{remark}

\begin{algorithm}
\caption{Numerical Algorithm For Solving $(\BT, \widetilde{\BT})$}
\label{alg: solve for de}
\begin{algorithmic}[1]
\STATE Set the parameter $z \in \C - \R_+$, the correlation matrix $\BC$, and the iteration counter $t = 0$. Initialize the matrices $\BT^{(0)} = -\frac{1}{z}\BI_N$ and $\widetilde{\BT}^{(0)} = -\frac{1}{z}\BI_M$.
\REPEAT
    \STATE $t \leftarrow t + 1$.
    \STATE $\BT^{(t)} \leftarrow -\frac{1}{z} [\BI_N + \widetilde{\cL}(\widetilde{\BT}^{(t-1)}) ]^{-1}$.
    \STATE $\widetilde{\BT}^{(t)} \leftarrow -\frac{1}{z} [ \BI_M + {\cL}(\BT^{(t-1)}) ]^{-1}$.
\UNTIL Convergence
\STATE \textbf{Output:} $(\BT^{(t)}, \widetilde{\BT}^{(t)})$.
\end{algorithmic}
\end{algorithm}

Before we present the main result, we introduce some additional notation for ease of presentation. Let $\cK: \C^{NM \times NM} \mapsto \C^{M \times M}$ be the map defined by $[\cK(\BM)]_{i, j} = \norm{\BM_{ij}}$, and let $\cW(\BB)$ denote the bandwidth of matrix $\BB$ such that $[\BB]_{i, j} = 0$ for $|i - j| > \cW(\BB)$. The following theorem demonstrates the deterministic approximation of the trace form of the resolvent.
\begin{theorem}[Deterministic approximation for $\BQ(z)$]
\label{thm: de for the resolvent}
    Assume Assumptions \ref{assumpt: A1}-\ref{assumpt: A3} hold. Let $\BA$ be a deterministic matrix with uniformly bounded spectral norm, i.e., $\limsup_{N \to \infty} \norm{\BA} < \infty$. Then, for any $z \in \C - \R_+$, we have 
    \begin{equation}
    \label{eq: resolvent converge}
        \Tr \BA\E \BQ(z) = \Tr \BA \BT(z) + \mathcal{O}\left(\frac{{\Psi}(z) \mathsf{E}_N}{\dz^4}\right),
    \end{equation}
    where 
    \begin{align}
    \label{eq: convergence E_N}
    \mathsf{E}_N = \frac{1}{N}{\min\left\{\sqrt{\ell},  1 + \inf_{\BU \in \mathbb{U}(M)}\norm{\cK(\BC_U)}^2\min(N, \cW[\cK(\BC_U)]+1)\right\}}
    \end{align}
    with $\BC_U = (\BU^\top \otimes \BI_N)\BC(\BU^* \otimes \BI_N)$ and 
    $\ell$ being the dimension of the vector space $\mathrm{Span}\{\BC_{i j}\}_{i, j=1}^M$. Here, $\Psi(z)$ is given by $\Psi(z) = \sum_{9 \leq j \leq 12, 0 \leq i \leq j} \frac{|z|^i}{\dz^j}$ with $\dz = \dist(z, \R_+)$.
\end{theorem}

\textit{Proof:} The proof of Theorem \ref{thm: de for the resolvent} is given in Appendix \ref{app: proof of thm de for the resolvent}. \qed
\par
\begin{remark}
In each term of $\Psi(z)$, the exponent of $|z|$ in the numerator does not exceed that in the denominator $\dz$. The motivation for introducing this polynomial is to ensure that when $x < 0$ and $|x|$ is sufficiently large, $\frac{\Psi(x)}{x^4}$ is integrable over the unbounded interval.
\end{remark}


\begin{remark}[Unitary invariance of $\ell$]
The dimension $\ell$ is an intrinsic parameter determined jointly by the correlation matrix $\BC$ and the matricization of $\Bh$ into $\BH\in\C^{N\times M}$. In particular, $\ell$ is invariant under unitary changes of basis at both the transceiver sides.  Let
$
\widetilde{\BH}=\BU_1\BH\BU_2$ with $
\BU_1\in\mathbb U(N) $ and $
\BU_2\in\mathbb U(M)
$. Then,
the covariance matrix of $\Vecm(\widetilde{\BH})$ is given by
$\widetilde{\BC} = (\BU_2^\top\otimes\BU_1) \BC (\BU_2^*\otimes\BU_1^H) $. Consequently, we have 
$ \dim\Span\{\widetilde{\BC}_{ij}\}_{i,j=1}^M = \ell $,
which implies that $\ell$ is invariant under rotations.
\end{remark}

 Theorem \ref{thm: de for the resolvent} can be used to characterize the asymptotic behavior of the ESD $\widehat{F}^{N}$. In fact, by Proposition~\ref{prop: exist unique DE}, the normalized trace $m_N(z) = \frac{1}{N}\Tr \BT(z)$ is the Stieltjes transform of a deterministic probability measure $F^{N}$, namely,
\begin{align}
  m_N(z)
    =
    \int_{\R_+}\frac{F^{N}(\dd\lambda)}{\lambda-z},
    \qquad z\in\C - \R_+. \label{eq: LSD definition}
\end{align}
Combining Theorem~\ref{thm: de for the resolvent} with Lemma~\ref{lem: variance control}, we obtain, for every fixed $z\in\C - \R_+$,
$
    \widehat{m}_N(z)
    -
    m_N(z)
     \to 0
$ almost surely.
Consequently, we have that $\widehat{F}^N - F^N$ converges weakly to 0 almost surely. Hereafter, we refer to $F^N$ as the limiting spectral distribution (LSD).  The following corollary further quantifies the associated convergence of the spectral functionals.

\begin{corollary} 
\label{coro: functional convergence}    
Let $\varphi \in \cC_c^\infty(\R, \R)$, the set of smooth real-valued functions with compact support in $\R$. Define the random LSS $\widehat{\cH}_N(\varphi) = \int \varphi (\lambda) \widehat{F}^N(\dd \lambda)$ and the deterministic spectral functional ${\cH}_N(\varphi) = \int \varphi (\lambda) {F}^N(\dd \lambda)$. Then, we have 
\begin{align}
    \E \left[ \widehat{\mathcal{H}}_N(\varphi) \right] - \cH_N(\varphi) = \mathcal{O}\left( \frac{\mathsf{E}_N}{N}\right), \quad \Var\left( \widehat{\cH}_N(\varphi) \right) = \mathcal{O}\left(\frac{\max_{x \in \mathrm{supp}(F^N) } | x\varphi'(x)^2|}{N^2} + \frac{\mathsf{E}_N}{N^3}\right).
\end{align}
\end{corollary}
\textit{Proof:}
The proof follows by combining Theorem~\ref{thm: de for the resolvent} with arguments analogous to those in \cite{AblaNoEigen, Vallet2012Subspace}. The details are therefore omitted. \qed

\begin{remark}[Asymptotic location of extreme eigenvalues]
Corollary~\ref{coro: functional convergence} can be used to characterize the asymptotic locations of the extreme eigenvalues of generally correlated Gaussian matrices. In particular, suppose that there exists some $\epsilon > 0$ such that
$
    \sum_{N \geq 1}N^{\epsilon}\mathsf{E}_N^2<\infty.
$
Then, it can be shown that no eigenvalues of $\BH\BH^H$ occur outside the limiting support $\operatorname{supp}(F^N)$ with probability one \cite{bai1998no, PAUL2009Noeigen, Bai2012Noeigen, Bai1999exact,AblaNoEigen, zhuang2026noeigen}. This no-eigenvalue property can be applied to multi-user MIMO communication to validate the invertibility of the channel Gram matrix, thereby ensuring the feasibility of zero-forcing precoding \cite{AblaNoEigen,zhuang2026noeigen}.
\end{remark}
In Theorem \ref{thm: de for the resolvent} and Corollary \ref{coro: functional convergence}, the parameter $\mathsf{E}_N$ is essential for the convergence rate. A primary observation is that $\frac{1}{N} \lesssim \mathsf{E}_N \lesssim 1$ since $\ell \leq \min(N^2, M^2)$. This indicates that $\mathsf{E}_N$  does not vanish in the worst case. However, as shown in Subsection~\ref{subsec: discusssion}, the order of $\mathsf{E}_N$ is small for many correlated models of $\BH$ in practice.

\subsection{Convergence of the EMI}
\label{Sec_Conv_EMI}
In this subsection, we apply the resolvent approximation established above to characterize the MI defined in \eqref{eq: Def of MI}, which is an LSS of $\BH\BH^H$ associated with
$
f(\lambda)=\log (1+ \lambda / \sigma^2)
$.
Corollary \ref{coro: functional convergence} yields a deterministic approximation for the EMI, while the following theorem provides a more explicit form.
\begin{theorem}
\label{thm: EMI approximation}
    Assume the same settings as in Theorem \ref{thm: de for the resolvent}. Then, we have the following deterministic approximation for the EMI 
    \begin{equation}
    \label{eq: de MI convergence rate}
        \E \left[ \mathcal{I}(\sigma^2) \right] = \overline{\cI}(\sigma^2) + \cO\left(\mathsf{E}_N \right),
    \end{equation}
    where 
    \begin{equation}
        \overline{\mathcal{I}}(\sigma^2) = \log\det\left(\BI_N + \widetilde{\cL}(\widetilde{\BT}(-\sigma^2))\right) + \log\det\left(\BI_M + \cL(\BT(-\sigma^2))\right) - \sigma^2 \innerprod{\cL(\BT(-\sigma^2))}{\widetilde{\BT}(-\sigma^2)}.
    \end{equation}
\end{theorem}
\textit{Proof:} The proof of Theorem \ref{thm: EMI approximation} is given in Appendix \ref{app: proof of thm EMI approximation}. \qed

In Section \ref{Sec: App comm}, we apply this theorem to several MIMO communication systems. To achieve the capacity, the covariance $\BP=\E[\Bx\Bx^H]$ of the transmitted signal needs to be designed subject to the power constraints.  Suppose the feasible region of $\BP$ is $\mathcal{P}$. The following remark provides a unified framework for the design of $\BP$.
\begin{remark}[Iterative algorithm for EMI maximization]
With transmit covariance matrix $\BP$, the MI is given by
$\cI(\BP;\sigma^2) = \log\det(\BI_N+\BH\BP\BH^H / \sigma^2) $,
where $\BH$ is defined in \eqref{eq: R channel model}. The EMI maximization problem is formulated as $\max_{\BP \in \cP} \E [\cI(\BP;\sigma^2)]$. Directly solving this problem is computationally challenging due to the expectation in the objective. However, the deterministic approximation in Theorem \ref{thm: EMI approximation} enables a more efficient design of $\BP$.
\par
By incorporating $\BP$ into $\BC$, the two operators are given by $ \cL_{P}(\BX) = \BP^{1/2}\cL(\BX)\BP^{1/2} $ and  $ \widetilde{\cL}_{P}(\BY)= \widetilde{\cL}(\BP^{1/2}\BY\BP^{1/2})$, where $(\cL, \widetilde{\cL})$ are defined in \eqref{eq: def of two operators}.
Let $(\BT_P(-\sigma^2),\widetilde{\BT}_P(-\sigma^2))$
denote the solution of \eqref{eq: DE systems} with $(\cL,\widetilde{\cL})$ replaced by $(\cL_{P},\widetilde{\cL}_{P})$. When $\mathsf{E}_N$ is small, the  EMI maximization problem can be approximated by the following deterministic optimization problem
\begin{align*}
    \mathrm{P1:} \qquad \max_{\BP \in \cP} ~ \overline{\cI}(\BP;\sigma^2) &=
    \log\det\left( \BI_N+ \widetilde{\cL}_P ( \widetilde{\BT}_P(-\sigma^2))\right) + \log\det\left( \BI_M+\BP \cL(\BT_P(-\sigma^2))\right) \\
    &- 
    \sigma^2
    \innerprod{
        \cL_{P}(\BT_P(-\sigma^2))
    }{
        \widetilde{\BT}_P(-\sigma^2)
    }. \numberthis
\end{align*}
Problem $\mathrm{P1}$ can be solved using various optimization methods \cite{boyd2004convex} or machine learning techniques. We provide one approach here. Since the operator $\cL$ is independent of $\BP$, $\mathrm{P1}$ can be solved by an iterative method. In particular, at the $t$-th iteration, we fix the solution $(\BT^{(t)}_P, \widetilde{\BT}^{(t)}_P)$ and update $\BP$ by solving
\begin{align}
\BP^{(t+1)} =  \underset{\BP \in \cP}{\arg\max} ~ \log\det\left(
        \BI_M +
        \BP {\cL} ({\BT}^{(t)}_P(-\sigma^2))\right). \label{eq: P t_1}
\end{align}
We note that when the feasible set is $\cP = \{\BP \succeq \mathbf{0}: \Tr \BP \leq M P_{\mathrm{tot}}\}$, where $P_{\mathrm{tot}}$ denotes the total power budget, \eqref{eq: P t_1} can be solved using the standard water-filling method. Next, Algorithm~\ref{alg: solve for de} is applied to compute $(\BT^{(t+1)}_P, \widetilde{\BT}^{(t+1)}_P)$. For many structured MIMO channels, this iterative procedure yields an asymptotically capacity-achieving transmit covariance \cite{Dumont2010rician,Couillet2011DE}.
\end{remark}

\subsection{Recovery of Existing Results for Structured Channel Models}
\label{subsec: discusssion}
In this subsection, we demonstrate the broad applicability of Theorems \ref{thm: de for the resolvent} and \ref{thm: EMI approximation} by showing how they recover existing deterministic approximations for various channel models. Specifically, by expressing the correlation of each MIMO channel in the form of $\mathbf{C}$, the two theorems directly yield deterministic approximations for the corresponding resolvent and EMI, together with their convergence rates $\mathsf{E}_N$.


\textit{1) Separable Correlated Model:}  In this case, the distribution of $\BH$ is given by 
\begin{align}
    \BH = \BD^{\frac{1}{2}} \BZ \widetilde{\BD}^{\frac{1}{2}}, \label{Eq_Model_Separable}
\end{align}
where deterministic matrices $\BD \in \C^{N \times N}$ and $\widetilde{\BD} \in \C^{M \times M}$ are Hermitian non-negative and the elements of $\BZ$ are i.i.d. complex Gaussian random variables with $Z_{1, 1} \sim \cC\cN(0, \frac{1}{M})$.  The convergence of the EMI with this model has been extensively studied in the context of MIMO communications \cite{Kermoal2002, Chen2002MIMOtit, Hachem2008ANewApproach} and statistics \cite{PAUL2009Noeigen, DingYang2021separable, Zhang2024nips}. In this setting, we have 
$\Bh = \Vecm(\BH) = (\widetilde{\BD}^{\frac{\top}{2}} \otimes \BD^{\frac{1}{2}}) \Vecm(\BZ)$ and the correlation matrix $\BC$ can be expressed as $\BC = \widetilde{\BD}^{\top} \otimes \BD$. Consequently, $\BC_{ij} = \widetilde{D}_{ji} \BD$ and the dimension $\ell = 1$. Under this setting and given that Assumptions \ref{assumpt: A1}-\ref{assumpt: A3} hold, we obtain $\mathsf{E}_N \asymp \frac{1}{N}$. We note this bound matches that in \cite[Theorem 3, Proposition 5]{Hachem2008ANewApproach}. 

\textit{2) Column-Independent Correlation Model:} For this model, the random matrix $\BH$ is given by
\begin{align}
    \BH = \begin{bmatrix}
        \BD^{\frac{1}{2}}_1 \Bz_1, \BD^{\frac{1}{2}}_2 \Bz_2, \ldots,  \BD^{\frac{1}{2}}_M \Bz_M
    \end{bmatrix}, 
\end{align}
where the $\BD_j$s are $N \times N$ Hermitian non-negative deterministic matrices and the $\Bz_j$s are i.i.d. Gaussian vectors with $\Bz_1 \sim \cC\cN(0, \BI_N / M)$. The column-independent model has rich applications in multi-user MIMO communication systems \cite{Adhikary2013, WagnerMISO2012,  AblaNoEigen, Zhuang2025Decentralized, zhuang2026noeigen} and multivariate statistics \cite{KAMMOUN2017Mestimator, Yin2020singular, louart2021spectral, MeiTianxing2023singular}. It can be verified that the correlation matrix is block diagonal and is given by $\BC = \mathrm{diag}(\BD_j: 1 \leq j \leq M)$. In this case, the dimension $\ell \leq M$ and the bandwidth $\cW(\cK(\BC)) = 0$. When the spectral norms of $\BD_j$s are uniformly bounded, we have 
$
    \mathsf{E}_N \leq \frac{\min(  \sqrt{N}, 1 + \max_j \norm{\BD_j}^2)}{N} \asymp \frac{1}{N}.
$
We note that the order of this error bound coincides with that in \cite[Proposition 3]{AblaNoEigen}. 

\textit{3) Mixed-Separable Correlation Model:} In cell-free mMIMO or distributed MIMO systems \cite{Wen2013NonGaussian, Couillet2011DE, zhang2013jsacMIMO}, the channel spatial correlations between different antenna sets may vary. As a result, the channel $\BH$ can be written as a block matrix $\BH = (\BH_{ab})$ with $a \leq B_N$, $b \leq B_M$, and $\BH_{ab} \in \mathbb{C}^{N_a \times M_b}$. The blocks are independent and the distribution of each block follows the same form as that in \eqref{Eq_Model_Separable}, i.e., $\BH_{ab} = \BD^{\frac{1}{2}}_{ab} \BZ_{ab} \widetilde{\BD}^{\frac{1}{2}}_{ab}$. Assume that the dimensions $B_N$ and $B_M$ are fixed, and $N_a$, $M_b$
grow with $N$ at the same ratio, i.e., $0 < \inf_a \lim\inf_N \frac{N_a}{N} \leq \sup_a\limsup_N \frac{N_a}{N} < \infty $ and $0 < \inf_b \lim\inf_N \frac{M_b}{N} \leq \sup_b \limsup_N \frac{M_b}{N} < \infty $. Moreover, the spectral norms of $\BD_{ab}$ and $\widetilde{\BD}_{ab}$ are uniformly bounded. Under these conditions, it can be verified that 
the dimension $\ell \leq B_NB_M$. Consequently, we obtain $\mathsf{E}_N \asymp \frac{1}{N}$, and \eqref{eq: resolvent converge} and \eqref{eq: de MI convergence rate} reduce to the zero-mean results in Eqs. (17) and (20) of \cite{zhang2013jsacMIMO}, respectively.

\textit{4) Additive-Separable Correlation Model:}
In MIMO-OFDM systems, the transceivers are connected through multiple paths with distinct delays. Suppose that there are $L$ paths. Then, the equivalent channel can be expressed as
$
    \BH
    =
    \sum_{r=1}^{L}
    \BH_r,
$
where $\BH_r$ denotes the channel matrix of the $r$-th path and follows the Kronecker Rayleigh model \cite{Bol2002OFDM, Moustakas2007, Mckay2008variance, Dupuy2011freqselec}, i.e., 
$
    \BH_r
    =
    \BD_r^{1/2}\BZ_r\widetilde{\BD}_r^{1/2}
$,
which is the same as \eqref{Eq_Model_Separable}. We note that the LSS of $\BH\BH^H$ has also been used to characterize performance metrics of task-oriented MIMO systems \cite{li2026task}. Under such circumstances, we have
$
    \BC
    =
    \sum_{r=1}^{L}
    \widetilde{\BD}_r^{\top}\otimes\BD_r.
$ Hence, if $L$ is fixed and
$
    \limsup_{N}
    \max_{r}
    \{
        \|\BD_r\|,
        \|\widetilde{\BD}_r\|
    \}
    <\infty,
$
there holds
$
    \mathsf{E}_N\asymp\frac{1}{N}$.
This convergence rate recovers those obtained in \cite[Propositions 1, 2]{Dupuy2011freqselec} and \cite[Theorems 1, 2]{li2026task} as special cases.

\section{Applications to Specific MIMO Systems}
\label{Sec: App comm}


In this section, we apply Theorem \ref{thm: EMI approximation} to derive deterministic approximations for the EMI of distributed and heterogeneous MIMO systems. These systems represent emerging trends toward distributed antenna deployments and heterogeneous hardware architectures, resulting in complex spatial correlations. Consequently, existing EMI analyses based on conventional structured channel models are not applicable.


\subsection{ Downlink Distributed MIMO System}
\label{sec: dist MIMO system}
Consider a distributed MIMO system, where the BS with $M$ antennas transmits to one user equipped with $N$ antennas. Specifically, there are $L + 1$ antenna sets at the BS, labeled as $\mathrm{BS}_0$, $\mathrm{BS}_1$, $\ldots$, $\mathrm{BS}_{L}$ and they are all connected to a central coordinator. We assume that $\mathrm{BS}_0$ is equipped with $M_0$ antennas, while each of the remaining antenna sets is equipped with a single antenna. Then, the number of antennas satisfies $M_0 + L = M$ and the received signal at the user is given by
\begin{align}
    \By = \BH_0 \Bx_0 + \sum_{j=1}^L \Bh_j x_j + \Bn,
\end{align}
where $\Bx_0 \in \C^{M_0}$ and $x_j \in \C$, $j \geq 1$, denote the zero-mean transmitted signals at $\mathrm{BS}_0$ and $\mathrm{BS}_j$, $j \geq 1$, respectively. Denote the concatenated signal as $\Bx = [\Bx_0^\top, x_1, \ldots, x_L]^\top \in \C^{M}$ and $\BP$ as the transmit covariance matrix with $\E [\Bx\Bx^H] = \BP$. Let $\BH_0$ denote the channel from \(\mathrm{BS}_0\) to the user, and let \(\Bh_j\) denote the channel from \(\mathrm{BS}_j\) to the user, for \(1\le j\le L\).  $\Bn \in \C^N$ represents the AWGN at the user with $\Bn \sim \cC\cN(0, \sigma^2 \BI_N)$. Assume the channels are in Rayleigh fading and the Kronecker channel model is adopted with
\begin{align}
    \BH_0 = \BD_0^{\frac{1}{2}} \BX_0 \widetilde{\BD}_0^{\frac{1}{2}}, ~~ \Bh_j = \BD_j^{\frac{1}{2}} \Bx_j, ~~ 1 \leq j \leq L, 
\end{align}
where $\BD_j \in \C^{N \times N}$, $j \geq 0$, and $\widetilde{\BD}_0 \in \C^{M_0 \times M_0}$ represent the spatial correlations and are Hermitian non-negative. The elements of $\BX_0$ and $\Bx_j$, $j \geq 1$ are i.i.d. complex Gaussian random variables with mean zero and variance $\frac{1}{M}$. 
\par
Suppose the BS only has statistical CSI. Define the concatenated channel matrix $\BH = [\BH_0, \Bh_1, \ldots, \Bh_L] \in \C^{N \times M}$. Then, the MI between the BS and the user is given by 
\begin{align}
    \cI(\BP; \sigma^2) = \log\det\left( \BI_N + \frac{1}{\sigma^2}\BH \BP \BH^H\right).
\end{align}
By applying Theorem \ref{thm: EMI approximation}, we obtain the following deterministic approximation for the EMI.
\begin{corollary} 
\label{coro: dist MIMO EMI approx}
Assume that the dimensions $N$, $M_0$, and $L$ approach infinity with $0 < \liminf \frac{M_0}{N} \leq \limsup \frac{M_0}{N} < \infty$, $0 < \liminf \frac{L}{N} \leq \limsup \frac{L}{N} < \infty$ and matrices $\widetilde{\BD}_0$, $\{\BD_j\}_{0 \leq j \leq L}$, and $\BP$ satisfy Assumptions \ref{assumpt: A2} and \ref{assumpt: A3}. Then, we have 
\begin{align}
    \E \left[ \cI(\BP; \sigma^2) \right] - \overline{\cI}(\BP; \sigma^2) = \mathcal{O}\left( N^{-\frac{1}{2}}\right),
\end{align}
where 
\begin{align*}
    \overline{\cI}(\BP; \sigma^2) &= \log\det\left[ \sigma^{-2} \BT^{-1} \right] + \log\det\left[ \sigma^{-2} \widetilde{\BT}^{-1}\right] - \sigma^2 M\delta_0 \widetilde{\delta}_0 - \sigma^2\sum_{i = 1}^L \delta_i \widetilde{\delta}_i, \numberthis \label{eq: de for MI dist MIMO system}
\end{align*}
with $\{\delta_i\}_{0 \leq i \leq L}$  and $\{\widetilde{\delta}_j\}_{0 \leq j \leq L}$ being positive solutions to the following system of equations, and matrices $\BT$ and $\widetilde{\BT}$ defined accordingly
\begin{align}
    \left\{
\begin{array}{rcl}
    \delta_i &=& \frac{\Tr \BD_i \BT}{M} ,~~~ 0 \leq i \leq L, \\
    \widetilde{\delta}_0 &=& \frac{ \Tr \widetilde{\BD}_0 [\BP^{\frac{1}{2}}\widetilde{\BT}\BP^{\frac{1}{2}}]_{[1:M_0, 1:M_0]}}{M}, \\
    \widetilde{\delta}_j &=& [\BP^{\frac{1}{2}}\widetilde{\BT} \BP^{\frac{1}{2}}]_{M_0 + j, M_0 + j}, ~~~ 1 \leq j \leq L, \\
    \BT &=& \frac{1}{\sigma^2}\left[\BI_N + \BD_0 \widetilde{\delta}_0 + \frac{1}{M} \sum_{j = 1}^L \BD_j \widetilde{\delta}_j \right]^{-1}, \\
    \widetilde{\BT} &=& \frac{1}{\sigma^2} \left[\BI_M + \BP^{\frac{1}{2}} \mathrm{blkdiag}\left(\delta_0 \widetilde{\BD}_0, \delta_1, \delta_2, \ldots, \delta_L \right) \BP^{\frac{1}{2}}\right]^{-1}.
\end{array} \right.
\end{align}
\end{corollary}
\begin{remark}
    Sufficient conditions for matrices $\widetilde{\BD}_0$, $\{\BD_j\}_{0 \leq j \leq L}$, and $\BP$ to satisfy Assumptions \ref{assumpt: A2} and \ref{assumpt: A3} are 
    \begin{align}
    \limsup_{N \to \infty}  \max \{ \max_j \lnorm{\BD_j}, \lnorm{\widetilde{\BD}_0}, \norm{\BP} \}< \infty, ~~  \liminf_{N \to \infty}   \min \{ \min_j \lambda_{\min}({\BD_j}), \lambda_{\min}(\widetilde{\BD}_0), \lambda_{\min}(\BP)\} > 0.
    \end{align}
\end{remark}
\begin{remark}
In \cite{Wen2013NonGaussian, zhang2013jsacMIMO, Couillet2011DE}, the number of antennas in each antenna set $\mathrm{BS}_k$ is assumed to grow with the system dimensions, while the total number of antenna sets, $L+1$, remains fixed. In contrast, we consider a regime in which both the number of antennas in $\mathrm{BS}_0$ and the total number of antenna sets can grow large. This regime captures a general class of distributed MIMO systems in which antennas are partitioned in a highly imbalanced manner across BS antenna sets. In the presence of spatial correlation, such imbalance typically leads to performance degradation \cite{Zhuang2025Decentralized}.
\end{remark}
\subsection{Uplink Heterogeneous MIMO System}
\label{sec: uplink heter MIMO}
Consider a user with $M$ antennas that transmits to a distributed BS comprising two antenna sets, $\mathrm{BS}_1$ (with $N_1$ antennas) and $\mathrm{BS}_2$ (with $N_2$ antennas). Denote the zero-mean transmit signal of the user as $\Bx \in \C^M$. Define the transmit covariance as $\E [\Bx\Bx^H] = \BP$.  The channel between the user and $\mathrm{BS}_k$ is denoted by $\BH_k \in \C^{N_k \times M}$, $k = 1, 2$. Then, 
the received signal at $\mathrm{BS}_k$ is given by
\begin{align}
    \By_k = \BH_k \Bx + \Bn_k,
\end{align}
with $\By_k \in \C^{N_k}$ and $\Bn_k \sim \cC\cN(0, \sigma^2 \BI_{N_k})$ being the AWGN at the receiver. 
\par
Due to the different locations of the two antenna sets, the scattering environments between each set and the user may differ. As a result, the statistical properties of $\BH_1$  and $\BH_2$ can be different. Accordingly, we assume that the distributions of $\BH_1$ and $\BH_2$ are given by
\begin{align}
    \BH_1 = \BS_r \left( \BSigma \odot \BX_1 \right) \BS_t^H, ~~ \BH_2 = \BD^{\frac{1}{2}} \BX_2 \widetilde{\BD}^{\frac{1}{2}},
\end{align}
where $\BS_r \in \C^{N_1 \times N_1}$ and $\BS_t \in \C^{M \times M}$ are deterministic matrices. The deterministic matrix $\BSigma = (\sigma_{ij}) \in \R_+^{N_1 \times M}$ denotes the amplitude profile and has non-negative elements. $\BD \in \C^{N_2 \times N_2}$ and $\widetilde{\BD} \in \C^{M \times M}$ are deterministic Hermitian non-negative matrices. The elements of $\BX_1$ and $\BX_2$ are i.i.d. complex Gaussian random variables with mean zero and variance $\frac{1}{M}$. We note that $\BH_2$ follows the Kronecker model and $\BH_1$ corresponds to the Weichselberger model when $\BS_r$ and  $\BS_t$ are unitary matrices \cite{Weichs2006channel, Luannan2016freemimo}. This scenario can arise when $\mathrm{BS}_1$ employs a holographic antenna array with mutual coupling (MC) \cite{Pizzo2022Holo, Zhang2025HolographicMIMO, Zhang2025HolographicMC}, while  $\mathrm{BS}_2$ is a uniform linear array.
Assume the user only has access to statistical CSI. Let $\BH = [\BH_1^\top, \BH_2^\top]^\top$. Then, the MI between the BS and the user is given by
\begin{align}
    \cI(\BP; \sigma^2) = \log\det\left( \BI_N + \frac{1}{\sigma^2}\BH \BP \BH^H\right) .
\end{align}
Denote $N = N_1 + N_2$, $\BR_a = \diag(\sigma_{ia}^2; 1 \leq i \leq N_1)$, and  $\widetilde{\BR}_b = \diag(\sigma_{bj}^2; 1 \leq j \leq M)$. By Theorem \ref{thm: EMI approximation}, we have the following deterministic approximation for $\E[\cI(\BP; \sigma^2)]$.
\begin{corollary}
\label{coro: heter MIMO}
Assume that the dimensions $N_1, N_2$ and $M$ approach infinity with $0 < \liminf \frac{N_k}{N} \leq \limsup \frac{N_k}{N} < \infty$, $k = 1, 2$, $0 < \liminf \frac{M}{N} \leq \limsup \frac{M}{N} < \infty$, and matrices $\BS_r$, $\BS_t$, $\BSigma$, $\BD$, $\widetilde{\BD}$, and $\BP$ satisfy Assumptions \ref{assumpt: A2} and \ref{assumpt: A3}. Then, we have 
\begin{align}
    \E \left[ \cI(\BP; \sigma^2) \right] - \overline{\cI}(\BP; \sigma^2) = \cO\left( N^{-\frac{1}{2}}\right),
\end{align}
where 
\begin{align*}
    \overline{\cI}(\BP; \sigma^2) &= \log\det\left[ \sigma^{-2}\BT_1^{-1} \right] + \log\det\left[ \sigma^{-2}\BT_2^{-1} \right] + \log\det\left[ \sigma^{-2}\widetilde{\BT}^{-1} \right]   \\
    &- M \sigma^2\delta_2 \widetilde{\delta}_2 - \sigma^2\sum_{i=1}^M \delta_{1i} [\BS_t^H \BP^{\frac{1}{2}}\widetilde{\BT} \BP^{\frac{1}{2}} \BS_t]_{i, i}. \numberthis
\end{align*}
with $\{\delta_{1i}\}_{1 \leq i \leq M}$, $\delta_2$, $\{\widetilde{\delta}_{1j}\}_{1 \leq j \leq N_1}$, and $\widetilde{\delta}_2$ being positive solutions of the following system of equations, and matrices $\BT_1$, $\BT_2$, and $\widetilde{\BT}$ defined accordingly
\begin{align}
\left\{
\begin{array}{rcl}
\delta_{1i} &=& \frac{\Tr \BR_i \BS_r^H \BT_1 \BS_r}{M} , ~~ 1 \leq i \leq M, \\
\delta_2 &=& \frac{ \Tr \BD \BT_2}{M}, \\
\widetilde{\delta}_{1j} &=& \frac{\Tr \widetilde{\BR}_j \BS_t^H \BP^{\frac{1}{2}}\widetilde{\BT} \BP^{\frac{1}{2}} \BS_t }{M}, ~~ 1 \leq j \leq N_1, \\
\widetilde{\delta}_2 &=& \frac{\Tr \widetilde{\BD} \BP^{\frac{1}{2}}\widetilde{\BT} \BP^{\frac{1}{2}} }{M}, \\
\BT_1 &=& \frac{1}{\sigma^2 }\left[ \BI_{N_1} + \BS_r \diag( \widetilde{\delta}_{1j}; j \leq N_1) \BS_r^H \right]^{-1}, \\
\BT_2 &=& \frac{1}{\sigma^2} \left[ \BI_{N_2} + \widetilde{\delta}_2 \BD \right]^{-1}, \\
\widetilde{\BT} &=& \frac{1}{\sigma^2} \left[ \BI_M + \BP^{\frac{1}{2}} \left( \delta_2 \widetilde{\BD} + \BS_t \diag(\delta_{1i}; i \leq M) \BS_t^H \right) \BP^{\frac{1}{2}} \right]^{-1}.
\end{array}
\right.
\end{align}
\end{corollary}
\begin{remark}
    Sufficient conditions for matrices $\BS_r$, $\BS_t$, $\BSigma$, $\BD$, $\widetilde{\BD}$, and $\BP$ to satisfy Assumptions \ref{assumpt: A2} and \ref{assumpt: A3} are 
    \begin{align}
    &\limsup_{N \to \infty} \max \{ \lnorm{\BS_r }, \lnorm{\BS_t}, \lnorm{\widetilde{\BD}}, \norm{\BD}, \norm{\BP}, \max_{i,j}\sigma_{ij}^2 \}< \infty,  \\
  &\liminf_{N \to \infty}  \min \{  \lambda_{\min}({\BS_r} \BS_r^H), \lambda_{\min}({\BS_t} \BS_t^H), \lambda_{\min}({\BD}), \lambda_{\min}(\widetilde{\BD}), \lambda_{\min}(\BP), \min_{i, j} \sigma_{ij}^2 \} > 0.
    \end{align}
\end{remark}

\begin{remark}
In contrast to the $\mathcal{O}(N^{-1})$ convergence rate for the mixed-separable correlated model discussed in Section \ref{subsec: discusssion}, the correlations $\BS_r$ and $\BS_t$, together with the amplitude profile $\BSigma$, induce a non-separable correlation structure. This additional correlation results in the slower convergence rate $\mathcal{O}(N^{-1/2})$.
\end{remark}

\section{Numerical Results}
\label{Sec: simulations}
In this section, we validate the accuracy of the theoretical results by simulations. In particular, we consider the following correlation model 
\begin{align}
    [\BCD(d_{\lambda}, \theta, \delta, N)]_{m, n} = \int_{-180}^{180}  \frac{\dd \phi}{\sqrt{2 \pi \delta^2}} e^{ 2 \pi \jmath (m-n) d_{\lambda} \sin\left( \frac{\pi \phi}{180}\right)  - \frac{(\phi - \theta)^2}{2 \delta^2} },
\end{align}
which is commonly used to model the spatial correlation for uniform linear arrays \cite{Moustakas2003titcorrelated}. Here, $m, n \in [N]$ denote the antenna indices, $d_\lambda$ represents the nearest neighbor antenna spacing (in wavelengths), and $\theta$ and $\delta$ are the mean angle and the root-mean-square angular spread, respectively, both measured in degrees.

\begin{figure}[t]
    \centering
    \includegraphics[width=3.8in]{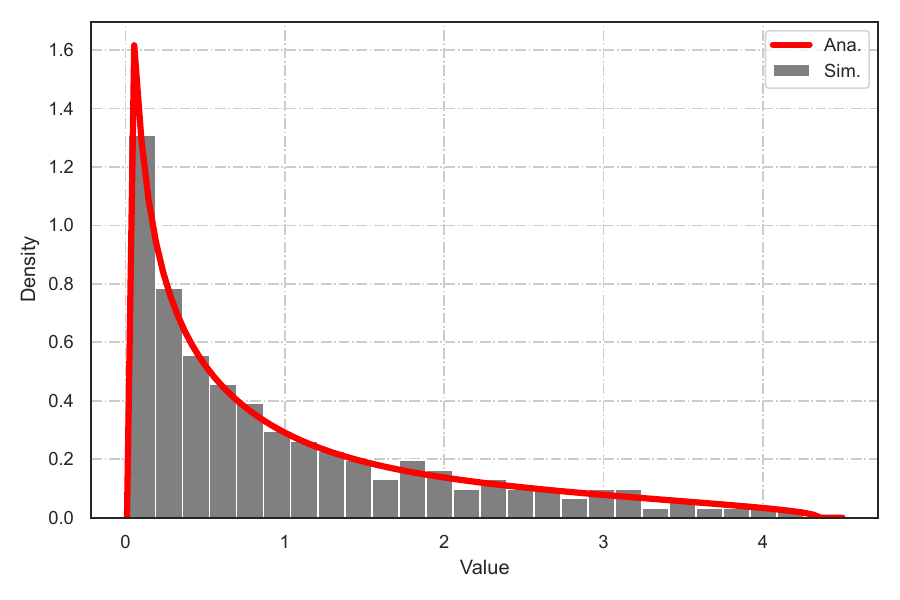}
    \caption{ESD versus LSD.}
    \label{fig: ESD vs LSD}
\end{figure}

\subsection{Validation of the ESD Approximation}
In Fig. \ref{fig: ESD vs LSD}, we compare the ESD with the corresponding LSD defined in \eqref{eq: LSD definition} of general Gaussian matrices. In particular, the density of the LSD can be numerically approximated by solving \eqref{eq: DE systems} using Algorithm \ref{alg: solve for de} and evaluating
$
\rho_{\mathrm{LSD}}(x) \approx \frac{1}{\pi}\Im m_N(x+\imath\varepsilon)
$
over the interval of interest, where $\varepsilon>0$ is sufficiently small. The dimensions are set as $N = 180$ and $M = 250$. Write $\BH = [\Bh_1, \ldots, \Bh_M]$. For \(i=1,\ldots,M/2\), the joint distribution of each column pair is specified by $[\Bh_i^\top, \Bh_{i + M / 2}^\top]^\top \sim \cC\cN(0, \BC_i / M)$, where $[\BC_i]_{k, \ell} = (0.3 + 0.2i / M)^{\abs{k - \ell}}$. It can be observed that the ESD of the covariance matrix $\BH\BH^H$ is close to the LSD, which validates the accuracy of Theorem \ref{thm: de for the resolvent}.

\begin{figure}[t]
    \centering
    \includegraphics[width=3.8in]{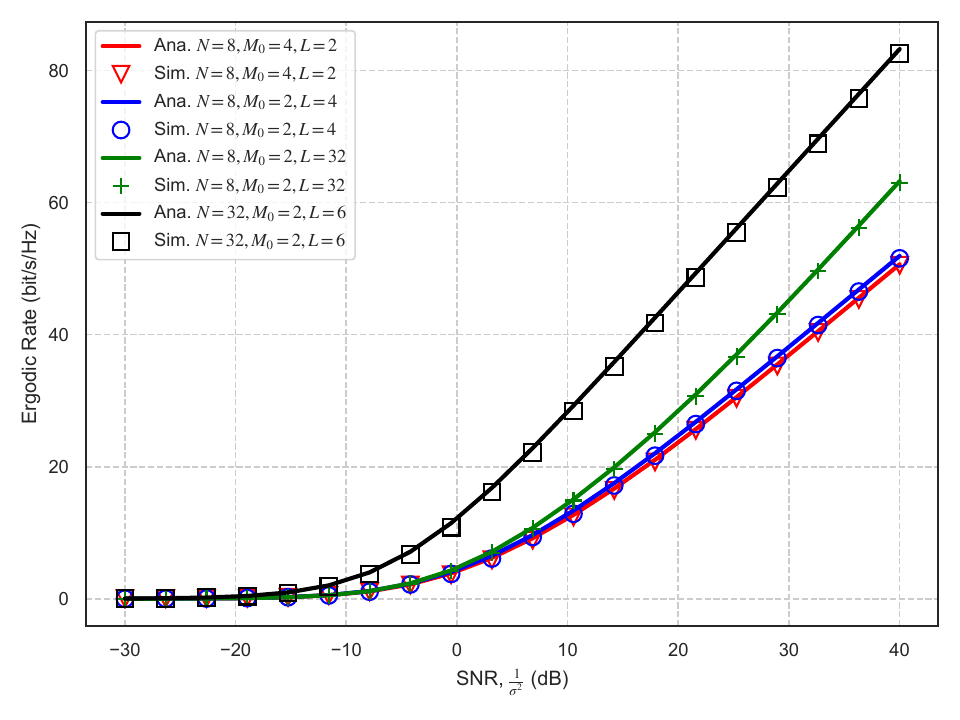}
    \caption{EMI versus SNR (Downlink Distributed MIMO System).}
    \label{fig: EMI SNR Dist MIMO}
\end{figure}

\subsection{EMI Approximation for Downlink Distributed MIMO} 

Fig. \ref{fig: EMI SNR Dist MIMO} illustrates the EMI performance at different SNRs in the downlink distributed MIMO system as discussed in Section \ref{sec: dist MIMO system}. The spatial correlation matrices are set as $\BD_0 = \BCD(\frac{1}{2}, 10, 20, N)$, $\widetilde{\BD}_0 = \BCD(\frac{1}{2}, 10, 20, M_0)$, and $\BD_j = \BCD(\frac{1}{2}, \theta_j,\delta_j, N)$ with $\theta_j = \frac{60j}{L}$ and $\delta_j = 5 + \frac{5j}{L}$. The transmit covariance matrix is set to $\BP = \BI_M$. The Monte Carlo simulation results are obtained by averaging over 500 independent trials. It can be observed that the deterministic approximation remains accurate even with small dimensions, thereby validating Corollary~\ref{coro: dist MIMO EMI approx}. Moreover, the EMI increases with the number of transmit antennas, demonstrating the performance gain provided by spatial DoF.


\begin{figure}[t]
    \centering
    \includegraphics[width=3.8in]{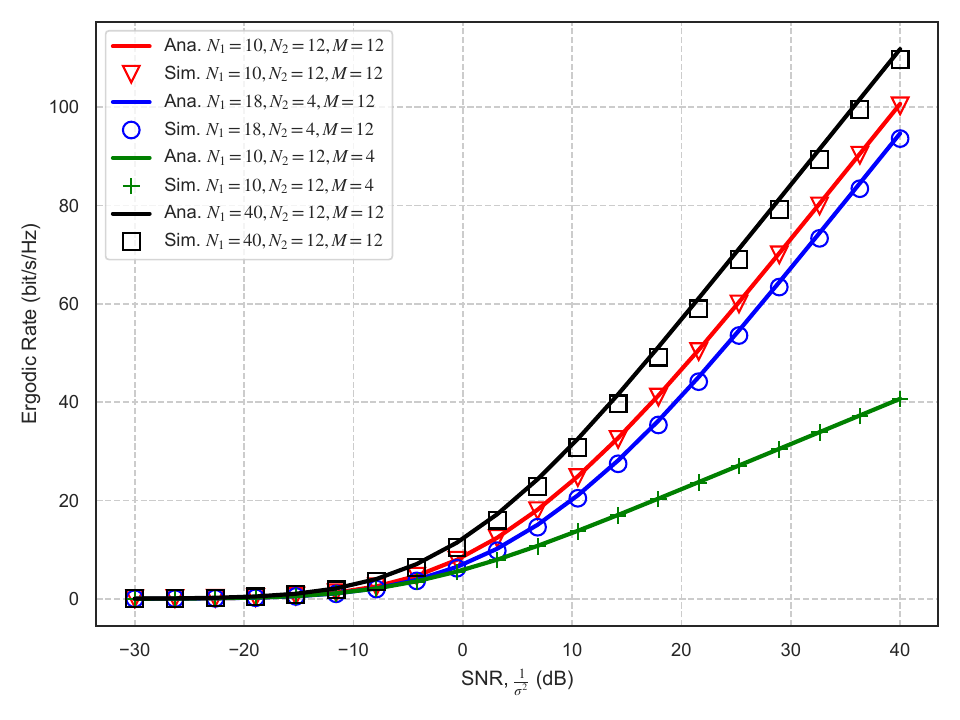}
    \caption{EMI versus SNR (Uplink Heterogeneous MIMO System).}
    \label{fig: EMI SNR Heter MIMO}
\end{figure}

\subsection{EMI Approximation for Uplink Heterogeneous MIMO} 
Fig. \ref{fig: EMI SNR Heter MIMO} illustrates the performance of EMI at different SNRs in the uplink heterogeneous MIMO system as discussed in Section \ref{sec: uplink heter MIMO}. The parameters are set to $\BS_r = \BS_{r, 2}^{\frac{1}{2}}$ and $\BS_t = \BS_{t, 2}^{\frac{1}{2}}$, where $[\BS_{r, 2}]_{i, j} = 0.9^{|i - j|}$ and $[\BS_{t, 2}]_{i, j} = 0.8^{|i-j|}$,  $\sigma_{ij} = \sqrt{0.95^{|i-j|} + 10^{-4}}$, $\BD = \BCD(\frac{1}{2}, 10, 50, N_2)$, $\widetilde{\BD} = \BCD(\frac{1}{2}, 50, 10, M)$, and $\BP = \BI_M$. The Monte Carlo simulation results are obtained by averaging over 500 independent trials. It can be observed that the deterministic approximation closely matches the Monte Carlo results, thereby validating Corollary~\ref{coro: heter MIMO}. When the total number of BS antennas and user antennas is fixed, allocating more antennas to $\mathrm{BS}_2$, whose channel follows the Kronecker model, improves the EMI performance. Moreover, increasing the number of user antennas provides an EMI gain.

\section{Conclusion}
\label{Sec: Conclu}

The increasingly diverse MIMO architectures give rise to more general and intricate channel correlation structures, rendering existing separable or block-structured model-based analyses   inadequate for characterizing the EMI. A central difficulty is that the fundamental self-consistent equations are matrix-valued, for which conventional scalar- or vector-based methods cannot be directly applied to analyze stability, which is  essential for establishing the convergence rate for the resolvent approximation. We addressed these challenges by establishing spectral norm control for general Gaussian matrices and developing a framework based on positive linear maps to analyze the fundamental equations. These developments enabled a unified deterministic approximation for the resolvent and EMI with an explicit convergence rate. The theoretical results were subsequently utilized to analyze the performance of downlink distributed MIMO and uplink heterogeneous MIMO systems, which cannot be tackled by existing methods. Numerical results demonstrated the accuracy of the approximations. More broadly, the proposed framework can support the analysis of emerging MIMO architectures: different antenna deployments and propagation environments can be directly incorporated through the covariance matrix $\BC$. Beyond performance evaluation, the proposed framework provides a theoretical basis for correlation-aware system design in modern MIMO systems. Future research directions include a second-order analysis of the LSSs for generally correlated Gaussian matrices and extensions of the proposed framework to non-Gaussian ensembles.


\appendices

\section{Mathematical Tools}
This section collects several mathematical tools that will be used in the proof of Theorem \ref{thm: de for the resolvent}.
\subsection{Results on Band Matrices}
\begin{lemma}
\label{lemm: band berstein bound}
    Let $\BA \in \C^{N \times N}$ be a band matrix with bandwidth $W$. Define the diagonal matrix with integer entries $\BCI = \diag(0, 1, 2, \ldots, N-1)$ and the map $\cB_I: \C^{N \times N} \to \C^{N \times N}$ as $\cB_I(\BX) = \BCI \BX - \BX \BCI$. Then, we have $\norm{\cB_I(\BA)} \leq W \norm{\BA}$.
\end{lemma}
\textit{Proof:} Denote the matrices $e^{\jmath \theta \BCI} = \diag(1, e^{\jmath \theta}, e^{2 \jmath \theta}, \ldots, e^{(N-1) \jmath \theta})$ and $\BA(\theta) = e^{\jmath \theta \BCI} \BA e^{- \jmath \theta \BCI}$. Then, we have $\norm{\BA(\theta)} = \norm{\BA}$ and $\frac{\dd \BA(\theta)}{ \dd \theta} = \jmath \cB_I(\BA(\theta))$. Write $\BA(\theta) = \sum_{k = -W}^W e^{\jmath k\theta}\BA_k$, where $\BA_k$ is a matrix with non-zero elements only on the $k$-th diagonal. Then, for any unit vectors $\Bx, \By \in \C^N$, we have
\begin{align}
   \abs{\Bx^H \cB_I(\BA(\theta)) \By} = \abs{ \Bx^H \frac{\dd \BA (\theta)}{\dd \theta} \By} = \abs{\sum_{k = -W}^W k e^{\jmath k \theta} \Bx^H \BA_k \By} \overset{(a)}{\leq} W  \sup_{\theta \in \R}\abs{ \sum_{k = -W}^W e^{\jmath k \theta} {\Bx^H \BA_k \By}} \leq W \norm{\BA},
\end{align}
where step $(a)$ follows from the Bernstein inequality for trigonometric polynomials.
By taking the supremum over all unit vectors $\Bx$ and $\By$ and letting $\theta = 0$ in the left-hand side (LHS) of the above inequality, we have $\norm{\cB_I(\BA)} \leq W \norm{\BA}$. \qed
\begin{lemma}
\label{lem: bound on banded Stieltjes transform}
    Let $\BL(z) \in \mathcal{S}_N(\R_+)$ with underlying measure $\Bmu_L$, i.e., $\BL(z) = \int_{\R_+} \frac{\Bmu_L(\dd \lambda)}{\lambda - z}$. Assume the bandwidth of $\Bmu_L$ is $W$, i.e., $[\Bmu_L]_{i, j} = 0$ for all $|i - j| > W$. Define the matrix $\BG(z) = -\frac{1}{z}[ \BI_N + \BL(z)]^{-1}$. Then,  for $z \in  \C - \R_+$,  we have 
    \begin{align}
        \abs{[\BG(z)]_{i, j}} \leq \frac{2}{\dz}\exp\left( -\frac{\mu_z|i-j|}{W + 1}\right),
    \end{align}
    where $\mu_z =  \min( \frac{\dz^2}{4|z|\norm{\Bmu_L(\R_+)}}, \frac{1}{2})$. If $\boldsymbol\mu_L(\mathbb R_+)=0$, we set $\mu_z=1/2$.
\end{lemma}
\textit{Proof:} The proof of Lemma \ref{lem: bound on banded Stieltjes transform} is given in Appendix \ref{app: lemma bound on banded Stieltjes transform}. \qed

\subsection{Results on Positive Linear Maps}
Let  $\cT: \C^{n \times n} \to \C^{m \times m}$ be a linear map. Define its operator norm as $\norm{\cT}_{op} = \sup_{\norm{\BX} = 1} \norm{\cT(\BX)}$. The linear map $\cT$ is called positive if $\cT(\mathbb{S}^n_+) \subset \mathbb{S}^m_+$, where $\mathbb{S}^n_+ = \{ \BX \in \C^{n \times n}: \BX \succeq 0 \}$ denotes the positive semidefinite cone. The following two lemmas describe properties of positive linear maps.
\begin{lemma}
\label{lem: bound on positive operator}
Let $\cT: \C^{N \times N} \mapsto \C^{N \times N}$ be a positive linear map. Assume there exist $\BX, \BY \succ 0$ such that $\BX = \cT(\BX) + \BY$. Then, we have
\begin{align}
    \rho(\cT) &\leq 1 - \frac{\lambda_{\min}(\BY)}{\lambda_{\max}(\BX)}, \label{eq: spectral radius positive linear operator} \\
    \norm{(\mathrm{Id}_N - \cT)^{-1}}_{op} &\leq \frac{\lambda_{\max}(\BX)} {\lambda_{\min}(\BY)}, \label{eq: spectral norm positive linear operator inverse}
\end{align}
where $\mathrm{Id}_N$ denotes the identity mapping on $\C^{N \times N}$. 
\end{lemma}
\textit{Proof:} From $\BY  \succeq \frac{\lambda_{\min}(\BY)}{\lambda_{\max}(\BX)} \BX$, we obtain $\cT(\BX) = \BX - \BY  \preceq (1 - \frac{\lambda_{\min}(\BY)}{\lambda_{\max}(\BX)}) \BX := \beta \BX$. Moreover, the inequality $\BX \succeq \lambda_{\min}(\BX) \BI_N$ implies $\lambda_{\min}(\BX) \cT(\BI_N) \preceq \cT(\BX)$. Applying the operator iteratively $k$ times gives
\begin{align}
    \lambda_{\min}(\BX) \cT^k(\BI_N) \preceq \cT^k(\BX) \preceq \beta^k \BX.
\end{align}
Hence, 
\begin{equation}
    \norm{\cT^k(\BI_N)}^{\frac{1}{k}} \leq \beta \left( \frac{\norm{\BX}}{\lambda_{\min}(\BX)} \right)^{\frac{1}{k}}. 
\end{equation}
According to the Russo–Dye Theorem \cite{russo1966note}, the spectral norm satisfies $\lnorm{\cT^k}_{op} = \lnorm{ \cT^k(\BI_N) } $. Then, by Gelfand's formula, we have 
\begin{align}
    \rho(\cT) = \lim_{k \to \infty} \norm{\cT^k}^{\frac{1}{k}}_{op} \leq \beta,
\end{align}
which establishes \eqref{eq: spectral radius positive linear operator}. For the second assertion, note that the map $\mathrm{Id}_N - \cT$ is invertible and $(\mathrm{Id}_N - \cT)^{-1}$ is positive, and 
\begin{align}
    \BX = (\mathrm{Id}_N - \cT)^{-1}\BY \succeq \lambda_{\min}(\BY)(\mathrm{Id}_N - \cT)^{-1}(\BI_N).
\end{align}
Taking the spectral norm on both sides of the above inequality proves \eqref{eq: spectral norm positive linear operator inverse}. \qed

\begin{lemma}[Choi's theorem]
\label{lem: Choi theorem}
    Let $\cT: \C^{N \times N} \mapsto \C^{M \times M}$ be a linear map and define $\boldsymbol{\Xi}_{\cT} = \sum_{i, j=1}^N \BE_{ ij} \otimes \cT(\BE_{ij})$, where $\BE_{ij} \in \C^{N \times N}$ denotes the matrix with $1$ at the $(i, j)$-th entry and $0$ otherwise. If $\BXi_\cT$ is Hermitian non-negative, then $\cT$ admits a representation of the form  
    \begin{align}
        \cT(\BX) = \sum_{k=1}^T \BV_k \BX \BV_k^H, 
    \end{align}
    where $T \leq MN$  and each $\BV_k \in \C^{M \times N}$.
\end{lemma}
\textit{Proof:} Since $\BXi_{\cT} \succeq 0$, it can be written as $\BXi_{\cT} = \sum_{j=1}^{NM} \Bv_j\Bv_j^H$. Let $\BP_j = \Be_{N, j}^\top \otimes \BI_M$. Then,  we have $\BP_a \Bv_j = \BV_j \Be_{N, a}$, where $\BV_j \in \C^{M \times N}$ satisfies $\Vecm(\BV_j) = \Bv_j$. Thus,
\begin{align}
    \cT(\BE_{ab}) = \BP_a  \BXi_{\cT} \BP_b^\top = \sum_{j=1}^{MN} \BP_a \Bv_j\Bv_j^H \BP_b^\top = \sum_{j=1}^{NM} \BV_{j}  \BE_{ab} \BV_j^H.
\end{align}
By linearity of $\cT$, Lemma \ref{lem: Choi theorem} is proved. \qed

\section{Proof of Theorem \ref{thm: bound norm}}
\label{app: proof of thm bound norm}
The proof of Theorem \ref{thm: bound norm} is based on an epsilon net argument \cite{tao2023topics}. To this end, we first establish the following probabilistic bound, which characterizes the random matrix $\BH$ as a random operator evaluated at a specific point $\Bu$. To simplify the notation, we denote $\BX = \sqrt{M}\BH$.
\begin{lemma}
\label{lem: bound inner prod}
    Let $\BX \in \mathbb{C}^{N \times M}$ be a random matrix and denote $\Bx = \Vecm(\BX)$. Assume $\Bx$ is a Gaussian vector such that $\Bx \sim \mathcal{CN}(0, \BC)$. Let $\Bu \in \mathbb{C}^N$ be a deterministic unit vector. Then, for any $t \geq \sqrt{6 \ln(2) \norm{\BC}}$, we have
    \begin{equation}
        \PP\left(\norm{\BX^H \Bu} \geq t\sqrt{M}\right) \leq\exp\left(-\frac{M t^2}{3 \norm{\BC}} \right). \label{eq: prob_inequal_Xu}
    \end{equation}
\end{lemma}
\textit{Proof:} If $\BC = \mathbf{0}_{NM \times NM}$ and $t > 0$, the desired result already holds. Without loss of generality, we assume $\BC \neq \mathbf{0}_{NM \times NM}$. Using the vectorization identity, we have 
\begin{align*}
    \norm{\BX^H \Bu}^2 = \Bu^H \BX\BX^H \Bu = \left[\Vecm(\BX)\right]^H (\BI_M \otimes \Bu\Bu^H) \Vecm(\BX) \overset{d}{=} \Bz^H \BC^{\frac{H}{2}} (\BI_M \otimes \Bu\Bu^H) \BC^{\frac{1}{2}} \Bz := \Bz^H \BG \Bz, \numberthis
\end{align*}
where the notation $\overset{d}{=}$ represents equality in distribution and $\Bz \sim \mathcal{CN}(0, \BI_{NM})$. By the identity $\rank(\BA \otimes \BB) = \rank(\BA) \rank(\BB)$, the rank of the matrix $\BG$ is not larger than $M$. Thus, the eigenvalue decomposition of $\BG$ can be given by $\BG = \BU \diag(\lambda_1, \ldots, \lambda_M, \\ 0, \ldots, 0) \BU^H$. Denote $\widetilde{\Bz} = \BU^H\Bz = [\widetilde{z}_1, \widetilde{z}_2, \ldots, \widetilde{z}_{NM}]^\top$. Due to the rotational invariance of the complex Gaussian distribution, we have $\widetilde{\Bz} \sim \mathcal{CN}(0, \BI_{NM})$ and thus
\begin{align*}
    \norm{\BX^H \Bu}^2 \overset{d}{=} \sum_{j=1}^M \lambda_j \abs{\widetilde{z}_j}^2 \overset{(a)}{\leq} \norm{\BC}\sum_{j=1}^M \abs{\widetilde{z}_j}^2, \numberthis
\end{align*}
where step $(a)$ follows from the identity $\lVert \BA \otimes \BB\rVert = \norm{\BA} \norm{\BB}$ and the bound $\max_j \lambda_j \leq \lVert\BC^{\frac{H}{2}} (\BI_M \otimes \Bu\Bu^H) \BC^{\frac{1}{2}} \rVert \leq \norm{\BC}$. Therefore, denoting $t' = \frac{t^2M}{\norm{\BC}}$, we have 
\begin{align*}
    &\PP \left( \norm{\BX^H \Bu} \geq t\sqrt{M} \right) \leq \PP \left( \sum_{j=1}^M \abs{\widetilde{z}_j}^2 \geq t' \right) = \PP \left( e^{\lambda \sum_{j=1}^M \abs{\widetilde{z}_j}^2} \geq e^{\lambda t'} \right) \leq e^{-\lambda t'} \prod_{j=1}^M \E \left( e^{\lambda \abs{\widetilde{z}_j}^2} \right) \\
    & = e^{-\lambda t'} \prod_{j=1}^M \E \left( e^{\lambda [\Re(\widetilde{z}_j)]^2} \right) \E \left( e^{\lambda [\Im(\widetilde{z}_j)]^2} \right) \overset{(a)}{=}  \frac{e^{-\lambda t'}}{(1 - \lambda)^M} \leq \exp\left[ -M \left( \frac{\lambda t^2}{\norm{\BC}} +\ln(1 - \lambda) \right) \right], ~~ \forall \lambda \in (0, 1). \numberthis
\end{align*}
Here, step $(a)$ follows from the fact that both the real part $\Re(\widetilde{z}_j)$ and the imaginary part  $\Im(\widetilde{z}_j)$ are independently distributed as $\mathcal{N}(0, \frac{1}{2})$.  By setting $\lambda = \frac{1}{2}$ in the above inequality, we have $\frac{t^2}{2 \norm{\BC}} + \ln(1 - \frac{1}{2}) \geq \frac{t^2}{3 \norm{\BC}}$ for any $t \geq \sqrt{6 \ln(2) \norm{\BC}}$ and \eqref{eq: prob_inequal_Xu} holds. Therefore, we have completed the proof. \qed
\par
In the following, we will prove Theorem \ref{thm: bound norm} using Lemma \ref{lem: bound inner prod}.  Let $\mathcal{N}_N$ be a maximal $\varepsilon$-net on the unit sphere $\mathbb{S}^{N-1} = \{ \Bx \in \C^N: \norm{\Bx} = 1\}$,  where $\varepsilon$ is a small constant to be chosen later. By definition, this means $\mathcal{N}_N \subset \mathbb{S}^{N-1} $ is a set that satisfies 
\begin{itemize}
    \item[1)] For any $\Bu \in \mathbb{S}^{N-1}$, there exists $\Bv \in \mathcal{N}_N$ such that $\norm{\Bu - \Bv} \leq \varepsilon$. 
    \item[2)] For any distinct $\Bv_1, \Bv_2 \in \mathcal{N}_N$, we have $\norm{\Bv_1 - \Bv_2} \geq {\varepsilon}$. 
\end{itemize}
Since $\mathbb{S}^{N-1}$ is compact and $\Bu \mapsto \lVert\BX^H \Bu \rVert$ is continuous, there exists $\Bu^{\star} \in \mathbb{S}^{N-1}$ such that $\norm{\BX} = \lVert\BX^H \Bu^{\star} \rVert$. By selecting $\Bv^\star \in \mathcal{N}_N$ such that $\norm{\Bu^{\star} - \Bv^\star} \leq \varepsilon$, we have 
\begin{align*}
 \norm{\BX^H \Bv^\star} = \norm{\BX^H \Bu^\star + \BX^H (\Bv^\star - \Bu^\star)} \geq  \norm{\BX^H \Bu^\star} - \norm{\BX^H (\Bv^\star - \Bu^\star)} \geq (1 - \varepsilon) \norm{\BX}, \numberthis
\end{align*}
since $\lVert{\BX^H (\Bv^\star - \Bu^\star)} \rVert \leq \lVert{\BX^H}\rVert \norm{ \Bv^\star - \Bu^\star} \leq \varepsilon \lVert{\BX^H}\rVert$. Consequently, the union bound gives, for any $a > 0$,
\begin{align*}
\label{eq: card ball}
    \PP \left( \norm{\BX} \geq a \right) \leq \PP \left( \max_{\Bv \in \mathcal{N}_N} \norm{\BX^H \Bv}  \geq (1 - \varepsilon)a \right) \leq \sum_{\Bv \in \mathcal{N}_N} \PP \left(  \norm{\BX^H \Bv}  \geq (1 - \varepsilon)a \right). \numberthis
\end{align*}
To bound the right-hand side (RHS) of \eqref{eq: card ball}, we need to estimate the cardinality of $\mathcal{N}_N$. Note that the open balls of radius $\varepsilon/2$ centered at distinct points of $\mathcal{N}_N$ are disjoint. Moreover, all these balls are contained in the enlarged ball $B(0, 1 + \varepsilon/2)$. By comparing the volumes of these sets, we obtain 
\begin{equation}
    \mathrm{Vol}(B(0, \varepsilon / 2))|\mathcal{N}_N| \leq \mathrm{Vol}(B(0, 1 + \varepsilon / 2)),
\end{equation}
which yields
\begin{equation}
    |\mathcal{N}_N| \leq \frac{\mathrm{Vol}(B(0, 1 + \varepsilon / 2))}{\mathrm{Vol}(B(0, \varepsilon / 2))} = (2 / \varepsilon + 1)^{2N}.
\end{equation}
By setting $\varepsilon = \frac{1}{2}$,  $a = t\sqrt{M} $ and requiring $t \geq 2\sqrt{6 \ln(2) \norm{\BC}}$ in \eqref{eq: card ball},  the following upper tail estimate is obtained by Lemma \ref{lem: bound inner prod}
\begin{equation}
\label{eq: prob bound on H}
    \PP (\norm{\BH} \geq t ) = \PP (\norm{\BX} \geq t \sqrt{M}) \leq (2 / \varepsilon + 1)^{2N} \exp\left( -\frac{(1 - \varepsilon)^2 M t^2}{3 \norm{\BC}} \right) = \exp\left( 2N \ln(5) - \frac{M t^2}{12\norm{\BC}}\right), 
\end{equation}
which implies 
\begin{align}
\PP(\lnorm{\BH} \geq t) \leq \exp(-NK_1t^2)  \label{eq: proba tail bound for norm H}
\end{align}
for $t \geq K$,
 where $K, K_1 > 0$ are constants independent of $N $ and $ M$. Consequently, we can obtain
\begin{equation}
    \sum_{N = 1}^{\infty}\PP (\norm{\BH} \geq K) \leq \sum_{N = 1}^{\infty} \exp(-NK_1K^2) < \infty,
\end{equation}
which establishes \eqref{eq: limsup prob bound H}. Next, using the identity $\E|X|^p = p\int_0^{\infty} \PP (|X| \geq t) t^{p-1} \dd t$, the following moment bound is obtained
\begin{align*}
\label{eq: expecation}
    \E \norm{\BH}^p &= p\int_0^{\infty} t^{p-1}\PP (\norm{\BH} \geq t)  \dd t \leq \int_0^K p t^{p-1} \dd t + \int_K^{\infty}pt^{p-1} \exp(-NK_1t^2) \dd t \leq K^p + \frac{p \Gamma(\frac{p}{2})}{2(NK_1)^{\frac{p}{2}}}, \numberthis
\end{align*}
where $\Gamma(z) = \int_0^{\infty} u^{z-1} e^{-u} \dd u$ denotes the standard Gamma function. Hence, we have $\sup_{N \geq 1} \E \lnorm{\BH}^p\leq K_p$, where $K_p > 0$ is a constant independent of $N$ and $M$. This completes the proof. \qed

\section{Proof of Proposition \ref{prop: exist unique DE}}
\label{app: proof of prop exist unique DE}
We prove the proposition by employing properties of the matrix-valued Stieltjes transform and Picard’s approximation method. To this end, we first give the formal definition of the class $\mathcal{S}_n(\R_+)$. A matrix-valued function $\BF(z) \in \C^{n \times n}$, $z \in \C-\R_+$, belongs to $\mathcal{S}_n(\R_+)$ if it admits the following integral representation
\begin{align}
    \BF(z) = \int_{\R_+} \frac{\Bmu(\dd \lambda)}{\lambda - z}, ~~ z \in \C- \R_+,
\end{align}
where $\Bmu = (\mu_{i,j})_{1 \leq i, j \leq n}: \mathscr{B}(\R_+) \mapsto \C^{n \times n}$ is a non-negative matrix-valued measure. That is, for any Borel set $B \subset \R_+$, the matrix $\Bmu(B) \succeq \mathbf{0}_{n \times n}$, and $\Tr \Bmu(\R_+) < \infty$. The following proposition provides a sufficient condition for a matrix-valued function to belong to the class $\cS_n(\R_+)$.
\begin{proposition}(\cite[Proposition 2.2]{hachem2007deterministic})
\label{prop: matrix-valued Stieltjes transform}
Let $\BF$ be an $n$-by-$n$ matrix-valued function on $\C_+$ satisfying:
\begin{itemize}
    \item[1]  The function $\BF$ is holomorphic on $z \in \C_+$.
    \item[2]  The matrices $\Im(\BF(z))$ and $\Im (z\BF(z))$ satisfy $\Im (\BF(z)) \succeq 0$ and $\Im (z\BF(z)) \succeq 0$, for any $z \in \C_+$.
    \item[3] $\lim_{y \to \infty} -\jmath y \BF(\jmath y) = \BA$, for some $\BA \succeq 0$.
\end{itemize}
Then, there exists a non-negative $n$-by-$n$ matrix-valued measure $\Bmu$ supported on $\R_+$ such that $\Bmu(\R_+) = \BA$ and $ \BF(z) = \int_{\R_+} \frac{\Bmu(\dd \lambda)}{\lambda - z}$. Moreover, $\BF(z)$ has an analytic continuation ${\BF}_a(z)$ on $\C - \R_+$ and ${\BF}_a(z) \in \cS_n(\R_+)$.
\end{proposition}
Since $\BF$ and its analytic continuation $\BF_a$
coincide on $\C_+$, we will not distinguish between them in what follows. The next two propositions show that certain transformations of elements in 
$\cS_n(\R_+)$ also belong to this class. These facts will be useful in the proof of Proposition \ref{prop: exist unique DE}.
\begin{proposition}
\label{prop: L in Sn}
    Let $\BG \in \cS_n(\R_+)$ and define $\BL(z) = -\frac{1}{z}(\BI_n + \BG(z))^{-1}$. Then, $ \BL \in \cS_n(\R_+)$. Moreover, the underlying measure of  $\BL$ satisfies $\Bmu(\R_+) = \BI_n$.
\end{proposition}
\textit{Proof: } It suffices to verify that $\BL$ satisfies the three properties stated in Proposition \ref{prop: matrix-valued Stieltjes transform}. Consider property $1$. Let $\BS = -z\BI - z\BG(z)$. For any $z \in \C_+$ and any nonzero $\Bx \in \C^n$,
\begin{align}
\label{eq: analytical of L on up half C}
    \Im(\Bx^H \BS \Bx) = \Bx^H \left(\frac{\BS - \BS^H}{2\jmath} \right) \Bx = - \Im(z) \norm{\Bx}^2- \Bx^H \Im(z \BG(z)) \Bx \leq -\Im(z) \norm{\Bx}^2.
\end{align}
Hence, for any $\Bx \neq 0$, $\Bx^H \BS \Bx \neq 0$ and thus $\BS$ is invertible. Consequently, the function $\BL = \BS^{-1}$ is holomorphic on $\C_+$ and property 1 is verified. Applying the resolvent identity, we can get 
\begin{align}
    \Im(\BL) = \frac{\BL - \BL^H}{2\jmath} = \frac{1}{2\jmath}\BL(z \BI_n - z^* \BI_n + z\BG - z^*\BG^H  )\BL^H = \Im(z) \BL\BL^H + \BL \Im(z\BG) \BL^H \succeq 0.
\end{align}
Similarly, $
    \Im(z\BL) = \frac{z\BL - z^*\BL^H}{2\jmath} = \frac{|z|^2}{2\jmath}\BL(\BG - \BG^H  )\BL^H = |z|^2\BL \Im(\BG) \BL^H \succeq 0$, 
which verifies property 2. Since $\BG \in \cS_n(\R_+)$, we have the limit $\lim_{y \to \infty} -\jmath y \BG(\jmath y)$ exists and $\lim_{y \to \infty} \norm{\BG(\jmath y)} = 0$. Therefore, 
\begin{align}
    \norm{\BI_n + \jmath y \BL(\jmath y)} \leq  \norm{\BG(\jmath y)} \norm{(\BI_n + \BG(\jmath y))^{-1}} \to 0,
\end{align}
which implies $\lim_{y \to \infty} - \jmath y \BL(\jmath y) = \BI_n$ and property 3 is verified. Therefore, Proposition \ref{prop: L in Sn} is proved. \qed
\begin{proposition}
\label{prop: cLG and wcLG in Sn}
    For $\BG \in \cS_N(\R_+)$ and $\widetilde{\BG} \in \cS_M(\R_+)$, we have the matrices $\cL(\BG(z)) \in \cS_M(\R_+)$ and $\widetilde{\cL}(\widetilde{\BG}(z)) \in \cS_N(\R_+)$. 
\end{proposition}
\textit{Proof:} The proof uses an argument similar to that of Proposition \ref{prop: L in Sn} and is omitted.
In addition, the following proposition provides a bound on the spectral norm of elements belonging to the class $\cS_n(\R_+)$.
\begin{proposition}
\label{prop: sp norm bound on elements in Sn}
    Let $\BG \in \cS_n(\R_+)$ and denote its underlying measure by $\Bmu$. Then, we have $
     \norm{\BG(z)}   \leq \frac{\norm{\Bmu(\R_+)}}{\dz}$ and $\norm{\Im \BG(z)} \leq \frac{|\Im(z)|\norm{\Bmu(\R_+)}}{\dz^2}$.
\end{proposition}
\textit{Proof:} The proof follows from the inequality $|\lambda - z| \geq \dz$ for $\lambda \geq 0$ together with a standard monotone class argument, and is therefore omitted. \qed
\subsection{Uniqueness}
We prove by contradiction. Assume that there exist two distinct solutions $(\BPhi, \widetilde{\BPhi})$ and $(\BPsi, \widetilde{\BPsi})$ solving \eqref{eq: DE systems} with $\BPhi, \BPsi \in \cS_N(\R_+)$ and $\widetilde{\BPhi}, \widetilde{\BPsi} \in \cS_M(\R_+)$. Define $\cD(z) = \max\{ \lnorm{ \BPhi(z) - \BPsi(z)}, \lnorm{\widetilde{\BPhi}(z) - \widetilde{\BPsi}(z)}\}$. Subtracting the two equations, applying the resolvent identity and Proposition \ref{prop: sp norm bound on elements in Sn}, we obtain, for $z\in \C_+$, 
\begin{align}
    \norm{\BPhi - \BPsi} = |z| \norm{\BPhi \left( \widetilde{\cL}(\widetilde{\BPhi} - \widetilde{\BPsi}) \right) \BPsi} \leq \frac{|z|\lnorm{\widetilde{\cL}}_{op}}{\dz^2} \norm{\widetilde{\BPhi} - \widetilde{\BPsi}}. \label{eq: norm on diff Phi Psi}
\end{align}
Similarly, we have $\lnorm{\widetilde{\BPhi} - \widetilde{\BPsi}} \leq \frac{|z| \norm{\cL}_{op}}{\dz^2} \lnorm{\BPhi - \BPsi}$. Hence, we obtain 
\begin{align}
    \cD(z) \leq \cE(z) \cD(z), ~~ \text{ where } \cE(z) = {|z|(\lnorm{\widetilde{\cL}}_{op} + \lnorm{\cL}_{op})}/{\dz^2}.
\end{align}
Therefore, on the domain $\mathbb{D} \subset \C_+$ where $|z| / \dz \leq 1$ and $\dz \geq 2(\lnorm{\widetilde{\cL}}_{op} + \lnorm{\cL}_{op})$, we have $\cE(z) \leq 0.5$ and $\cD(z) = 0$. Since every entry of $\BPhi$, $\BPsi$, $\widetilde{\BPhi}$, and $\widetilde{\BPsi}$ is holomorphic on $\C_+$, it follows that $(\BPhi, \widetilde{\BPhi}) = (\BPsi, \widetilde{\BPsi})$ for all $z \in \C - \R_+$. This completes the proof of uniqueness.
\subsection{Existence}
To prove existence, we construct the following iterative scheme.  Let $\BT^{(0)} = -\frac{1}{z} \BI_N \in \cS_N(\R_+)$ and $\widetilde{\BT}^{(0)} = -\frac{1}{z}\BI_M \in \cS_M(\R_+)$ be the initial iterates. For $t \geq 1$, let
\begin{align}
    \BT^{(t)}(z) &= -\frac{1}{z} \Big(\BI_N + \widetilde{\cL}(\widetilde{\BT}^{(t - 1)}(z)) \Big)^{-1}, \label{eq: iter for Tt}\\
    \widetilde{\BT}^{(t)}(z) &= -\frac{1}{z}\Big(\BI_M + \cL(\BT^{(t-1)}(z)) \Big)^{-1}. \label{eq: iter for wTt}
\end{align}
By Propositions \ref{prop: L in Sn} and \ref{prop: cLG and wcLG in Sn}, we have $\BT^{(t)} \in \cS_N(\R_+)$ and $\widetilde{\BT}^{(t)} \in \cS_M(\R_+)$ with underlying measures satisfying $\bm{\mu}_{\BT, t}(\R_+) = \BI_N$ and $\bm{\mu}_{\widetilde{\BT}, t}(\R_+) = \BI_M$ for all non-negative integers $t$. Denote $\cD^{(t)}(z) = \max\{ \lnorm{\BT^{(t)}(z) - \BT^{(t-1)}(z)},  \lnorm{\widetilde{\BT}^{(t)}(z) - \widetilde{\BT}^{(t-1)}(z) }  \}$. By a computation similar to that in \eqref{eq: norm on diff Phi Psi}, we have $\cD^{(t)}(z) \leq \cE(z) \cD^{(t-1)}(z)$. As a result, $\{ \BT^{(t)}(z) \}_{t \geq 0}$ and $\{ \widetilde{\BT}^{(t)}(z)\}_{t \geq 0}$ are Cauchy sequences on $\mathbb{D}$ with respect to the spectral norm. Denote their corresponding limits by $\BT(z)$ and $\widetilde{\BT}(z)$.  On the other hand, by Proposition \ref{prop: sp norm bound on elements in Sn}, the entries $[\BT^{(t)}]_{i, j}$ and $[\widetilde{\BT}^{(t)}]_{k, l}$ are bounded on every compact subset of $\C - \R_+$. Therefore, for each set of indices $i, j, k$, and $l$, the families $\{ [\BT^{(t)}]_{i, j} \}_{t \geq 0}$ and $\{ [\widetilde{\BT}^{(t)}]_{k, l} \}_{t \geq 0}$ are normal families. By the normal family theorem \cite{rudin1987real}, every subsequence of $\{ [\BT^{(t)}]_{i, j} \}_{t \geq 0}$ ($\{ [\widetilde{\BT}^{(t)}]_{k, l} \}_{t \geq 0}$) admits a further subsequence that converges to a holomorphic function. Because these limits agree with $[\BT]_{i, j}$ ($[\widetilde{\BT}]_{k, l}$) on $\mathbb{D}$, it follows that 
\begin{align}
    \lim_{t \to \infty }\BT^{(t)}(z) =  \BT(z), ~~  \lim_{t \to \infty} \widetilde{\BT}^{(t)}(z) = \widetilde{\BT}(z), ~~ z \in \C - \R_+,
\end{align}
where $\BT$ and $\widetilde{\BT}$ are holomorphic on $\C- \R_+$.
\par
Next, we prove that $\BT \in \cS_N(\R_+)$ and $\widetilde{\BT} \in \cS_M(\R_+)$ by verifying the 3 properties of Proposition \ref{prop: matrix-valued Stieltjes transform}. Property 1 has already been established. Property 2 follows from $\Im(\BT^{(t)}(z)), \Im(z\BT^{(t)}(z)) \succeq \mathbf{0}_{N \times N}$ and $\Im(\widetilde{\BT}^{(t)}(z)), \Im(z\widetilde{\BT}^{(t)}(z)) \succeq \mathbf{0}_{M \times M}$, which hold for every $t \geq 0$ and $z \in \C_+$. Passing to the limit preserves the inequalities. By an argument similar to that in \eqref{eq: analytical of L on up half C}, we have that matrices $-\frac{1}{z} (\BI_N + \widetilde{\cL}(\widetilde{\BT}(z)))^{-1}$ and $-\frac{1}{z} (\BI_M + {\cL}({\BT}(z)))^{-1}$ are holomorphic on $\C_+$. Hence, 
\begin{align*}
    \BT + \frac{1}{z} (\BI_N + \widetilde{\cL}(\widetilde{\BT}))^{-1} &= \BT - \BT^{(t)} + \frac{1}{z} (\BI_N + \widetilde{\cL}(\widetilde{\BT}))^{-1} - \frac{1}{z} (\BI_N + \widetilde{\cL}(\widetilde{\BT}^{(t-1)}))^{-1} \\
    & = \BT - \BT^{(t)} + (\BI_N + \widetilde{\cL}(\widetilde{\BT}))^{-1} \widetilde{\cL}(\widetilde{\BT} - \widetilde{\BT}^{(t-1)}) \BT^{(t)}, ~~ z \in \C_+. \numberthis
\end{align*}
By taking $\lim_{t \to \infty}$ on both sides of the above equation and using the continuity of $\widetilde{\cL}$ (e.g., $\lim_n \widetilde{\cL} X_n  =  \widetilde{\cL} \lim_n X_n$), we obtain $\BT(z) = -\frac{1}{z} (\BI_N + \widetilde{\cL}\widetilde{\BT}(z))^{-1}$ for $z \in \C_+$. A similar argument on $\widetilde{\BT}(z)$ yields that \eqref{eq: DE systems} holds for $z \in \C_+$. Since $\lnorm{\BT^{(t)}}, \lnorm{\widetilde{\BT}^{(t)}} \leq 1 / \Im(z)$ for $ t \geq 1$ and $z \in \C_+$, we have $\lnorm{\BT}, \lnorm{\widetilde{\BT}} \leq 1 / \Im(z)$. Therefore, 
\begin{align}
\label{eq: pass y limit to T}
    \lim_{y \to \infty }\BI_N + \jmath y \BT(\jmath y) = \lim_{y \to \infty } \widetilde{\cL}(\widetilde{\BT}(\jmath y)) \left(\BI_N + \widetilde{\cL}(\widetilde{\BT}(\jmath y)) \right)^{-1} = 0.
\end{align}
Similarly, we have $\lim_{y \to \infty } -\jmath y \widetilde{\BT}(\jmath y) = \BI_M$, which verifies property 3 in Proposition \ref{prop: matrix-valued Stieltjes transform}. Therefore, we have completed the proof of Proposition \ref{prop: exist unique DE}. \qed

\section{Proof of Theorem \ref{thm: de for the resolvent}}
\label{app: proof of thm de for the resolvent}
The proof of Theorem \ref{thm: de for the resolvent} consists of two parts: a stochastic step and a deterministic step. In the stochastic step, we employ Gaussian tools, i.e., the integration by parts formula \cite{Hachem2008ANewApproach} and the Poincaré-Nash inequality \cite{Hachem2008ANewApproach}, to establish the convergence of the resolvent $\BQ(z)$ to an intermediate quantity $\BT_\alpha(z)$, which is explicitly defined by
\begin{align}
    \widetilde{\BT}_\alpha(z) = -\frac{1}{z}\left(\BI_M + \cL[\E \BQ(z)] \right)^{-1}, ~~ \BT_\alpha(z) = -\frac{1}{z} \left(\BI_N +\widetilde{\mathcal{L}}[\widetilde{\BT}_\alpha (z)] \right)^{-1}.
\end{align}
This intermediate object is defined via the expectation of the resolvent but is more analytically tractable. In the deterministic step, we then show that $\BT_\alpha(z)$ converges to the solution $\BT(z)$ of the fixed-point equation. 
Recall $\dz = \dist(z, \R_+)$. For these two steps, the following spectral norm bounds on $\BT_\alpha(z)$ and $\widetilde{\BT}_\alpha(z)$ are useful, which is a direct application of Proposition \ref{prop: sp norm bound on elements in Sn}
\begin{align}
    \norm{\BT_\alpha(z)} \leq \frac{1}{\dz}, ~~ \norm{\widetilde{\BT}_\alpha(z)} \leq \frac{1}{\dz}. \label{eq: sp norm bound on T Talpha}
\end{align}
For simplicity, we will omit the variable $z$  if there is no confusion, and restore it whenever necessary. 
\subsection{Stochastic Step}
By the integration by parts formula and the derivative identity $\partial (\BX^{-1}) = - \BX^{-1} (\partial \BX) \BX^{-1}$, we have 
\begin{align*}
    \E Q_{i,k} H_{k,l} H_{j, s}^* &= \sum_{a, b} \E [H_{k, l} H_{a, b}^*] \E \frac{\partial Q_{i,k} H_{j, s}^*}{\partial H_{a, b}^*} = \sum_{a, b} \frac{[\BC_{lb}]_{k, a}}{M} \E\left[ -[\BQ\BH]_{i, b} Q_{a, k} H_{j, s}^* + Q_{i, k} \delta(a-j)\delta(b-s)\right] \\
    & = \sum_b \E \left[ -\frac{1}{M}[\BQ \BH]_{i, b}  [\BC_{lb}\BQ]_{k,k} H_{j,s}^*\right] + \frac{1}{M} \E Q_{i, k} [\BC_{ls}]_{k, j}. \numberthis \label{eq: Qik_Hkl}
\end{align*} 
By summing over the subscript $k$ in \eqref{eq: Qik_Hkl}, we can get 
\begin{align*}
    \E [\BQ\BH]_{i,l} H_{j, s}^* = \E \left[ -[\BQ \BH \cL(\BQ)]_{i, l}   H_{j,s}^*\right] + \frac{1}{M} \E  [\BQ\BC_{ls}]_{i, j}. \numberthis \label{eq: QH_il}
\end{align*} 
By writing $\BQ = \E \BQ +  \mathring \BQ$, we immediately get $\E [\BQ\BH (\BI_M + \cL[ \E \BQ ])]_{i,l} H_{j, s}^* = \E [ -[\BQ \BH \cL(  \mathring{\BQ})]_{i, l}   H_{j,s}^*] + \frac{1}{M} \E  [\BQ\BC_{ls}]_{i, j}$.  Multiplying both sides by $-z[\widetilde{\BT}_\alpha ]_{l, s}$ and summing over the subscripts $l$ and $s$ yield
\begin{align*}
    \E [\BQ\BH\BH^H]_{i, j} = \E \left[ z\left[\BQ \BH \cL(\mathring \BQ) \widetilde{\BT}_\alpha\BH^H\right]_{i, j} \right] -z  \E  \left[\BQ \widetilde{\cL}(\widetilde{\BT}_\alpha) \right]_{i, j}. \numberthis \label{eq: QHH_H}
\end{align*} 
The LHS of \eqref{eq: QHH_H} thus can be handled by the resolvent identity. In particular, by $\BQ\BH\BH^H = z\BQ + \BI_N$, we have
\begin{align}
    \E \left[\BQ \BT_\alpha^{-1}\right]_{i, j} = \delta(i-j) -z\E \left[\BQ \BH \cL(\mathring \BQ) \widetilde{\BT}_\alpha\BH^H\right]_{i, j}. \label{eq: QT}
\end{align}
Multiplying both sides of the preceding identity by $[\BT_\alpha\BA]_{j, i}$ and summing over the subscripts $i$ and $j$, the following is obtained
\begin{align}
     \Tr (\E\BQ - \BT_\alpha) \BA =  -z\E 
    \Tr \left[ \BQ \BH \cL(\mathring \BQ) \widetilde{\BT}_\alpha\BH^H \BT_\alpha \BA \right] = \varepsilon. \label{eq: QA}
\end{align}
Thus, we have established the relationship between $\E \BQ$ and $\BT_\alpha$ in the linear form. Ideally, the term $\varepsilon$ should be an infinitesimal quantity, and in the following, we will evaluate its order. We observe that $\varepsilon$ consists of a sum of covariance terms like $\mathrm{Cov}([\cL(\BQ)]_{i, j}, X_{i, j})$. Therefore, it is crucial to bound the moments of the trace forms for the resolvent, which is provided in the following lemma.
\begin{lemma} [Central Moment Control for the Resolvent]
\label{lem: variance control}
Assume Assumptions \ref{assumpt: A1} and \ref{assumpt: A2} hold. Let $\BA \in \C^{N \times N}$ and $\BB \in \C^{M \times M}$ be two deterministic matrices. Then, for any given integer $p > 0$ and $z \in \mathbb{C} - \R_+$, we have
\begin{align}
    &\E \abs{\Tr \BA \BQ(z) - \E \Tr \BA \BQ(z) }^p  = \mathcal{O}_p\left( \frac{\norm{\BA}_F^p}{N^{\frac{p}{2}} \mathsf{d}_z^{2p}}\right). \label{eq: moment bound on trace AQ} \\
    &\E \abs{ \Tr \BA \BQ(z) \BH \BB \BH^H - \E\Tr \BA \BQ(z) \BH \BB \BH^H }^p = \mathcal{O}_p\left( \frac{\min\left\{\norm{\BA} \norm{\BB}_F,  \norm{\BA}_F \norm{\BB}\right\}^p}{N^{\frac{p}{2}}} \left( \frac{1}{\dz^p} + \frac{1}{\dz^{2p}} \right) \right). \label{eq: moment bound on trace AQHbHH}
\end{align}
\end{lemma}
\textit{Proof:} The proof of Lemma \ref{lem: variance control} is given in Appendix \ref{app: proof of lemma variance control}. \qed
\par
Let $ \{ \widetilde{\BC}_1, \dots, \widetilde{\BC}_\ell \} $ be an orthonormal basis of the vector space $\mathrm{Span} \{\BC_{ij} \}_{i,j=1}^M $ in the matrix space $ \mathbb{C}^{N \times N} $, equipped with the inner product $ \langle \BX, \BY \rangle = \Tr(\BX\BY^H) $. Then, the matrix $ \BC $ can be expressed by $ \BC = \sum_{j=1}^\ell \BD_j \otimes \widetilde{\BC}_j $. Thus, the map $\cL$ can be represented as 
$
    \cL(\BX) = \sum_{j=1}^\ell \BD_j^\top \frac{\Tr \widetilde{\BC}_j \BX}{M}
$.
As a result, we have 
\begin{align*}
    &\E \Tr \left[ \BQ \BH \cL(\mathring{\BQ}) \widetilde{\BT}_\alpha\BH^H \BT_\alpha \BA \right] = \sum_{j=1}^\ell \E \left[ \frac{1}{M} {\Tr (\widetilde{\BC}_j \mathring{\BQ}})  \Tr (\BQ\BH \BD_j^\top \widetilde{\BT}_\alpha\BH^H \BT_\alpha \BA) \right]. \numberthis \label{eq: var control}
\end{align*}
By Lemma \ref{lem: variance control}, inequality \eqref{eq: sp norm bound on T Talpha}, and the Cauchy-Schwarz inequality $\E|xy| \leq \E^{\frac{1}{2}}|x|^2 \E^{\frac{1}{2}}|y|^2$, we have 
\begin{align*}
    |\varepsilon| &\leq \frac{|z|}{M} \sum_{j=1}^\ell \Var^{\frac{1}{2}}[\Tr \widetilde{\BC}_j \BQ ]\Var^{\frac{1}{2}}[\Tr  (\BT_\alpha \BA \BQ\BH\BD_j^\top \widetilde{\BT}_\alpha\BH^H)] \\
    &\leq  \cO(1) \frac{|z|}{M} \sum_{j=1}^\ell \frac{\sqrt{\Tr (\widetilde{\BC}_j \widetilde{\BC}_j^H) \norm{\BA}^2\Tr( \widetilde{\BT}_\alpha^H \BD_j^* \BD_j^\top  \widetilde{\BT}_\alpha) }}{N} \left( \frac{1}{\dz^4} + \frac{1}{\dz^5} \right)  \\
    &\overset{(a)}{\leq} \cO(1) \frac{|z|\sqrt{\ell \sum_{j=1}^\ell \Tr \BD_j^* \BD_j^\top}}{MN} \left(  \frac{1}{\dz^5} + \frac{1}{\dz^6} \right) \overset{(b)}{=}\mathcal{O} \left( \frac{\sqrt{\ell}}{N} \left( \frac{|z|}{\dz^5} + \frac{|z|}{\dz^6} \right) \right) , \numberthis \label{eq: pre bound ell}
\end{align*}
where in step $(a)$ we use the inequalities $|\Tr \BX \BY| \leq  \norm{\BX} \Tr \BY$ for Hermitian non-negative matrix $\BY$, and $\sum_{j=1}^n \sqrt{a_j} \leq \sqrt{n(\sum_{j=1}^n a_j)}$. Step $(b)$ follows by the inequality $\sum_{j=1}^\ell \Tr \BD_j^* \BD_j^\top = \sum_{a, b}\sum_{j} |[\BD_j]_{a, b}|^2 = \sum_{a, b} \Tr \BC_{ab} \BC_{ab}^H = \Tr \BC^2 \leq NM \norm{\BC}^2$.
\par
The error bound \eqref{eq: pre bound ell} implies that when $\ell \ll N$, i.e., the set of correlation matrices $\{\BC_{ij}\}_{i, j=1}^M$ is low-rank, the RHS of \eqref{eq: QA} becomes negligible. However, in the worst-case scenario $\ell \approx N^2$, $\varepsilon$ becomes an $\cO(1)$ term. A more refined analysis is therefore needed. To this end, we apply Gaussian tools to the error term $\E \Tr [ \BQ \BH \cL(\mathring{\BQ}) \widetilde{\BT}_\alpha\BH^H \BT_\alpha \BA ]$ to derive another bound.  By applying the integration by parts formula to the term  $\E Q_{i,k}H_{k, l}[\cL(\mathring\BQ)]_{s, m} H_{j, q}^*$, we obtain
\begin{align*}
    &\E Q_{i,k}H_{k, l} [\cL(\mathring\BQ)]_{s, m}  H_{j, q}^* = \sum_{a, b} \E[H_{k, l} H^*_{a, b}] \E \frac{\partial  (Q_{i,k}[\cL(\mathring\BQ)]_{s, m}  H_{j, q}^*)}{\partial H_{a, b}^*} \\
    &= \sum_{a, b} \frac{[\BC_{lb}]_{k, a}}{M} \E \left[ -[\BQ \BH]_{i, b} Q_{a,k}  [\cL(\mathring\BQ)]_{s, m} H_{j, q}^* - \frac{1}{M}Q_{i, k} [\BQ \BC_{ms}\BQ\BH]_{a, b} H_{j, q}^* + Q_{i,k} [\cL(\mathring\BQ)]_{s, m}  \delta(a-j)\delta(b-q)\right]. \numberthis
\end{align*}
By summing over the subscript $k$, we have 
\begin{align*}
    &\E [\BQ\BH]_{i,l} [\cL(\mathring\BQ)]_{s, m}  H_{j, q}^* = \E \left[ -[\BQ \BH \cL(\BQ)]_{i, l} [\cL(\mathring\BQ)]_{s, m} H_{j, q}^* - \sum_{b=1}^M \frac{1}{M^2} [\BQ\BC_{lb}\BQ \BC_{ms}\BQ\BH]_{i, b} H_{j, q}^* + \frac{[\BQ \BC_{lq}]_{i,j}[\cL(\mathring\BQ)]_{s, m} }{M}\right]. \numberthis
\end{align*}
Similar to \eqref{eq: QHH_H}, by writing $\cL(\BQ) = \cL(\E \BQ) + \cL(\mathring{\BQ})$, multiplying both sides by $-z[\widetilde{\BT}_\alpha]_{l, s}  [\widetilde{\BT}_\alpha]_{m, q}[\BT_\alpha \BA]_{j, i}$, and summing over the subscripts $m, q, l, s, i,$ and $j$, we obtain
\begin{align*}
    &\E \Tr [\BQ\BH \cL(\mathring{\BQ}) \widetilde{\BT}_\alpha \BH^H \BT_\alpha \BA] = z\E \left[ \Tr [\BQ \BH \cL(\mathring{\BQ}) \widetilde{\BT}_\alpha \cL(\mathring{\BQ}) \widetilde{\BT}_\alpha \BH^H \BT_\alpha \BA] \right]\\
    & +\sum_{s, l, b, m}\frac{z[\widetilde{\BT}_\alpha]_{l, s}}{M^2}\E \left[  [\widetilde{\BT}_\alpha \BH^H\BT_\alpha \BA\BQ\BC_{lb}\BQ \BC_{ms}\BQ\BH]_{m, b} \right] \\
    &-z \E \left[ \Tr \left[ \cL(\BT_\alpha \BA \BQ) \widetilde{\BT}_\alpha \cL(\mathring{\BQ}) \widetilde{\BT}_\alpha \right] \right] := \varepsilon_1 + \varepsilon_2 + \varepsilon_3. \numberthis
\end{align*}
In the following, we evaluate the error terms $\varepsilon_j$, $j=1,2,3$ separately. 
\par
Using the Cauchy–Schwarz inequalities $|\Tr \BX\BY| \leq (\Tr \BX \BX^H \Tr \BY\BY^H)^{\frac{1}{2}}$ and $\E |xy| \leq \E^{\frac{1}{2}} |x|^2 \E^\frac{1}{2} |y|^2 $, we have 
\begin{align*}
    |\varepsilon_1| &\leq |z| \E^{\frac{1}{2}} \left[ \Tr \cL(\mathring{\BQ})[\cL(\mathring{\BQ})]^H \right] \E^{\frac{1}{2}} \left[\Tr \widetilde{\BT}_\alpha \cL(\mathring{\BQ}) \widetilde{\BT}_\alpha \BH^H \BT_\alpha \BA\BQ\BH \BH^H \BQ^H \BA^H\BT_\alpha^H \BH \widetilde{\BT}_\alpha^H [\cL(\mathring{\BQ})]^H \widetilde{\BT}_\alpha^H\right] \\
    & \leq |z| \norm{\widetilde{\BT}_\alpha}^2  \norm{\BT_\alpha} \norm{\BA} \E^{\frac{1}{2}} \left[ \Tr \cL(\mathring{\BQ}) [\cL(\mathring{\BQ})]^H \right]\E^{\frac{1}{2}} \left[ \Tr \cL(\mathring{\BQ})[\cL(\mathring{\BQ})]^H \norm{\BH}^4\norm{\BQ}^2\right] \\
    &\overset{(a)}{=} \cO \left(\frac{|z| }{N^{\frac{1}{2}}\mathsf{d}_z^6} \right)\E^{\frac{1}{2}}\left[ \Tr \cL(\mathring{\BQ})[\cL(\mathring{\BQ})]^H \norm{\BH}^4\right], \numberthis \label{eq: bound on vepsilon1}
\end{align*}
where step $(a)$ follows from \eqref{eq: sp norm bound on T Talpha}, Lemma \ref{lem: variance control}, and the following estimate
\begin{align*}
    \E \left[ \Tr \cL(\mathring{\BQ})[\cL(\mathring{\BQ})]^H \right] = \sum_{i, j=1}^M \frac{1}{M^2}\Var(\Tr \BC_{ji} \BQ) = \cO\left( \frac{\sum_{i, j} \Tr \BC_{ji} \BC_{ji}^H}{M^2N \mathsf{d}_z^4} \right) =  \cO \left( \frac{ \Tr  \BC^2}{M^2N \mathsf{d}_z^4} \right) = \cO\left(\frac{1}{N\mathsf{d}_z^4}\right). \numberthis \label{eq: bound on TrLaLaH}
\end{align*}
As a result, to bound the error term 
$\varepsilon_1$, it is essential to establish a bound for $\E \Tr \cL(\mathring{\BQ}) [\cL(\mathring{\BQ})]^H \lnorm{\BH}^4$. By the Cauchy–Schwarz inequality, we have
\begin{align*}
    \E \left[ \Tr \cL(\mathring{\BQ}) [\cL(\mathring{\BQ})]^H \norm{\BH}^4 \right] \leq \E^{\frac{1}{2}} \left[ \Tr\cL(\mathring{\BQ}) [\cL(\mathring{\BQ})]^H \right]^2 \E^{\frac{1}{2}} \left[ \norm{\BH}^8 \right]. \numberthis \label{eq: Tr LQLQH4}
\end{align*}
Using Lemma \ref{lem: variance control}  together with another application of the Cauchy–Schwarz inequality, we have 
\begin{align*}
\E \left[ \Tr \cL(\mathring{\BQ})[\cL(\mathring{\BQ})]^H  \right]^2 &= \sum_{i, j, a, b = 1}^M \E \abs{[\cL(\mathring{\BQ})]_{i,j}[\cL(\mathring{\BQ})]_{a, b}}^2 \leq \sum_{i, j, a, b = 1}^M \E^{\frac{1}{2}} \left[\abs{[\cL(\mathring{\BQ})]_{i, j}}^4 \right]\E^{\frac{1}{2}} \left[\abs{[\cL(\mathring{\BQ})]_{a, b}}^4 \right] \\
&\leq  \frac{\cO \left(1\right)}{M^4}\sum_{i, j, a, b = 1}^M \frac{\Tr \BC_{ji}\BC_{ji}^H}{N \mathsf{d}_z^4} \frac{\Tr \BC_{ba}\BC_{ba}^H}{N \mathsf{d}_z^4}   =\cO \left(\frac{1}{M^4}\left( \frac{\Tr \BC^2}{N \mathsf{d}_z^4} \right)^2 \right) = \cO\left( \frac{1}{N^2 \mathsf{d}_z^8}\right). \numberthis \label{eq: Tr LQLQ2}
\end{align*}
Combining \eqref{eq: Tr LQLQ2}, \eqref{eq: Tr LQLQH4} and the norm bound in Theorem \ref{thm: bound norm}, we get 
\begin{align}
    |\varepsilon_1| =  \cO\left(\frac{|z|}{N \dz^8}\right).
\end{align}
\par
Next, we evaluate the order of $\varepsilon_2$. For a clearer presentation, we first introduce some notation.
 For an $NM$-by-$NM$ complex matrix $\BY$, define  $\BY^\Gamma$ as its block-wise transpose with $\BY^{\Gamma}[\Bvarsigma_i, \Bvarsigma_j] = \BY[\Bvarsigma_j, \Bvarsigma_i]$, for all $j, i \in \{1, 2, \ldots, M\}$.
Then, it can be verified that $\sum_{l, s} \BC_{lb} \BQ \BC_{ms} [\widetilde{\BT}_\alpha]_{l,s} = (\BC^\Gamma (\widetilde{\BT}_\alpha \otimes \BQ) \BC^\Gamma) [\Bvarsigma_b, \Bvarsigma_m]$. Hence, $|\varepsilon_2|$ can be bounded as 
\begin{align*}
    \abs{\varepsilon_2} &= \abs{\sum_{b, m}\frac{z}{M^2}\E \left[  [\widetilde{\BT}_\alpha \BH^H\BT_\alpha \BA\BQ(\BC^\Gamma (\widetilde{\BT}_\alpha \otimes \BQ) \BC^\Gamma) [\Bvarsigma_b, \Bvarsigma_m]\BQ\BH]_{m, b} \right]} \\
    & = \abs{ \frac{z}{M^2} \E \left[ \Bv^H_{Q, 1}(\BC^\Gamma (\widetilde{\BT}_\alpha \otimes \BQ) \BC^\Gamma)^\Gamma \Bv_{Q, 2} \right] } \leq \frac{|z|}{M^2} \E \left[ \norm{\Bv_{Q, 1}} \norm{\Bv_{Q, 2}} \norm{(\BC^\Gamma (\widetilde{\BT}_\alpha \otimes \BQ) \BC^\Gamma)^\Gamma} \right] \\
    & \leq \frac{|z|N \norm{\BA}}{M^2 \dz^4} \E \left[ \norm{\BH}^2 \norm{(\BC^\Gamma (\widetilde{\BT}_\alpha \otimes \BQ) \BC^\Gamma)^\Gamma} \right], \numberthis \label{eq: epsilon 2 bound}
\end{align*}
where $\Bv_{Q, 1} = \Vecm(\BQ^H\BA^H\BT_\alpha^H \BH \widetilde{\BT}_\alpha^H)$ and $\Bv_{Q, 2} = \Vecm(\BQ\BH)$. To bound $\varepsilon_2$, a crucial step is to bound $\lnorm{(\BC^\Gamma (\widetilde{\BT}_\alpha \otimes \BQ) \BC^\Gamma)^\Gamma}$. To this end, we define $\BXi_{C, Q} = \BC^\Gamma (\widetilde{\BT}_\alpha \otimes \BQ) \BC^\Gamma$. Let $\BU \in \U(M)$ be a given unitary matrix. Then, the matrix $\BXi_{C, Q}$ can be written as 
\begin{align*}
    \BXi_{C, Q} &= \BC^\Gamma (\BU \otimes \BI_N)  (\BU^H \otimes \BI_N) (\widetilde{\BT}_\alpha \otimes \BQ) (\BU \otimes \BI_N)  (\BU^H \otimes \BI_N) \BC^\Gamma \\
    &= \BC^\Gamma (\BU \otimes \BI_N)   (\BU^H\widetilde{\BT}_\alpha \BU \otimes \BQ)  (\BU^H \otimes \BI_N) \BC^\Gamma. \numberthis
\end{align*}
Using the identities $(\BV \otimes \BI_N) \BC^\Gamma  = [\BC (\BV^\top \otimes \BI_N)]^\Gamma$ and $\BC^\Gamma (\BV \otimes \BI_N) = [(\BV^\top \otimes \BI_N) \BC]^\Gamma$ for any $\BV \in \C^{M \times M}$, we obtain 
$ (\BU^\top \otimes \BI_N) \BXi_{C, Q}^\Gamma (\BU^* \otimes \BI_N) = (\BC_U^\Gamma (\widetilde{\BT}_{\alpha, U} \otimes \BQ) \BC_U^\Gamma)^\Gamma$, where $\BC_U = (\BU^\top \otimes \BI_N) \BC (\BU^* \otimes \BI_N)$ and $\widetilde{\BT}_{\alpha, U} = \BU^H \widetilde{\BT}_\alpha \BU$. Hence, we have 
\begin{align*}
   \norm{\BXi_{C, Q}^\Gamma} &= \norm{(\BC^\Gamma_U (\widetilde{\BT}_{\alpha, U} \otimes \BQ) \BC^\Gamma_U)^\Gamma} = \sup_{\norm{\Bx} = \norm{\By} = 1} \abs{\Bx^H(\BC^\Gamma_U (\widetilde{\BT}_{\alpha, U} \otimes \BQ) \BC^\Gamma_U)^\Gamma \By} \\
   &= \sup_{\norm{\Bx} = \norm{\By} = 1}  \abs{ \sum_{i, j=1}^M \Bx_i^H  \left[ (\BC^\Gamma_U(\widetilde{\BT}_{\alpha, U} \otimes \BQ) \BC^\Gamma_U)[\Bvarsigma_j, \Bvarsigma_i] \right] \By_j }\\
   &\leq  \norm{\BQ}\sup_{\norm{\Bx} = \norm{\By} = 1}  \sum_{i, j}\norm{\Bx_i} \norm{\By_j} \sum_{a, b} \norm{\BC_{U, aj}} \abs{[\widetilde{\BT}_{\alpha, U}]_{a, b}} \norm{\BC_{U, ib}} \\
   &\overset{(a)}{\leq} \frac{1}{\dz} \sup_{\norm{\Balpha} = \norm{\Bbeta} = 1} \abs{ \Balpha^\top  \cK(\BC_U) \abs{\widetilde{\BT}_{\alpha, U}} \cK(\BC_U) \Bbeta } \leq \frac{\norm{\cK(\BC_U)}^2 \norm{\abs{\widetilde{\BT}_{\alpha, U}}}}{\dz}. \numberthis \label{eq: CGamma T Q related norm bound}
\end{align*}
where we decompose the vector $\Bu = [\Bu_1^\top, \Bu_2^\top, \ldots, \Bu_M^\top]^\top$  with $\Bu_j \in \C^N$ for $j \in \{1, 2, \ldots, M\}$ and $\Bu \in \{\Bx, \By\}$. Step $(a)$ is due to the identity $\norm{\Bu}^2 = \sum_{j=1}^M \lnorm{\Bu_j}^2$. Next, we bound the spectral norm of $|\widetilde{\BT}_{\alpha, U}|$. Setting $\BL = \BU^H \cL(\E \BQ ) \BU$ in Lemma \ref{lem: bound on banded Stieltjes transform} and letting $W_C = \cW(\cK(\BC_U))$, we have $\norm{\Bmu_L(\R_+)} = \norm{\cL(\BI_N)} $ and $\mu_z = \min(\frac{\dz^2}{4|z| \norm{\cL(\BI_N)} }, \frac{1}{2})$. Using the inequalities $1 - e^{-x} \leq \min(1, x)$ and $e^{x} - 1 \geq x$ for $x \geq 0$, we have
\begin{align*}
    & \norm{\abs{\widetilde{\BT}_{\alpha, U}}}_{\infty} =\max_{1 \leq i \leq M} \sum_{j=1}^M \abs{[\widetilde{\BT}_{\alpha, U}]_{i, j}} \leq \frac{2}{\dz} \max_{1 \leq i \leq M} \sum_{j=1}^M \exp(-\frac{\mu_z}{W_C + 1} |i - j| ) \leq  \frac{2}{\dz} \left(1 + 2 \sum_{j=1}^M e^{-\frac{\mu_z}{W_C + 1} j}\right) \\
    &= \frac{2}{\dz} \left(1 + \frac{2(1 - e^{- M \frac{\mu_z}{W_C + 1}})}{e^{\frac{\mu_z}{W_C + 1}} - 1}\right) \leq \frac{2}{\dz}\left( 1 + \frac{2 \min(W_C + 1, M\mu_z)}{\mu_z}  \right) \leq \frac{K}{\dz}\left( \frac{|z|}{\dz^2} + 1\right)\min\left(W_C+1, N\right), \numberthis \label{eq: abs T alpha norm bound}
\end{align*}
where $K$ is a constant independent of $N$ and $M$.  The same estimates also hold for $\lnorm{|\widetilde{\BT}_{\alpha, U}|}_{1}$. Applying the inequality $\norm{\BX} \leq \sqrt{\norm{\BX}_1 \norm{\BX}_\infty}$, substituting the bounds from \eqref{eq: abs T alpha norm bound} into \eqref{eq: epsilon 2 bound} and \eqref{eq: CGamma T Q related norm bound}, and using Theorem \ref{thm: bound norm}, we obtain
\begin{align}
    |\varepsilon_2| = \mathcal{O}\left(  \frac{\norm{\cK(\BC_U)}^2 \min(\cW[\cK(\BC_U)] + 1, N)}{N} \left(\frac{|z|}{\dz^6} + \frac{|z|^2}{\dz^8} \right)\right).
\end{align}
\par
Finally, we evaluate $|\varepsilon_3|$.  Writing $\BQ = \mathring{\BQ} + \E \BQ$ and using the Cauchy-Schwarz inequality give
\begin{align}
    |\varepsilon_3| \leq \frac{|z|}{\dz^2} \E^{\frac{1}{2}} \left[ \Tr(\cL(\BT_\alpha \BA \mathring\BQ))(\cL(\BT_\alpha \BA \mathring\BQ))^H  \right] \E^{\frac{1}{2}} \left[ \Tr \cL(\mathring{\BQ})(\cL(\mathring{\BQ}))^H \right].
\end{align}
With the same method used in \eqref{eq: bound on TrLaLaH}, we can get 
\begin{align}
    |\varepsilon_3| = \mathcal{O} \left( \frac{|z|}{N \dz^7} \right).
\end{align}
Combining the bounds for $\varepsilon_j, j=1,2,3$, we have 
\begin{align}
    \abs{\varepsilon} \leq |z|(|\varepsilon_1| + |\varepsilon_2| + |\varepsilon_3|) \leq \mathcal{O}\left( \left(\frac{|z|^2}{\dz^6} + \frac{|z|^2}{\dz^7} + \frac{|z|^2 + |z|^3}{\dz^8}\right) \frac{ 1 + \norm{\cK(\BC_U)}^2 \min(\cW[\cK(\BC_U)] + 1, N) }{N}\right). \label{eq: detail bound ell}
\end{align}
Since $\BU \in \mathbb{U}(M)$ is arbitrary, we can take the infimum over $\mathbb{U}(M)$ outside the big-O notation. 
Combining the bounds in \eqref{eq: pre bound ell} and \eqref{eq: detail bound ell}, we have, for any $z \in \C - \R_+$ and $\BA$ with uniformly bounded spectral norm,  
\begin{align}
    \abs{\Tr \BA \BCE(z)} \leq \cO(P(z)\mathsf{E}_{N}), \label{eq: trace bound on E} 
\end{align}
where $P(z) = \frac{|z|}{\dz^5} + \frac{|z| + |z|^2}{\dz^6} + \frac{|z|^2}{\dz^7} + \frac{|z|^2 + |z|^3}{\dz^8}$  and $\BCE(z) = \E \BQ(z) - \BT_\alpha(z)$. Therefore, the stochastic step is completed.
\subsection{Deterministic Step}
In this step, we establish the connection between $\BT_\alpha$ and $\BT$ by controlling the linear functional associated with the difference $\BT_\alpha - \BT$. To this end, we first introduce the following four useful propositions for clarity of presentation.
\begin{proposition}
\label{prop: bound on norm eigenvalue T}
There exist  constants $K, \eta > 0$ (independent of $N, M$) such that 
    \begin{align}
        \min \left\{ \lambda_{\min}(\BT^H\BT), \lambda_{\min}(\BT_\alpha^H\BT_\alpha)  \right\} \geq \frac{K \dz^2}{(|z|^2 + \eta^2)^2},
    \end{align}
    for $z \in \C - \R_+$.
\end{proposition}
\textit{Proof:} The proof of  Proposition \ref{prop: bound on norm eigenvalue T} is given in Appendix \ref{app: proof of prop bound on norm eigenvalue T}. \qed 
\begin{proposition}
\label{prop: representation of L cL}
    The linear maps $\cL$ and $\widetilde{\cL}$ admit the representation
    \begin{align}
        \cL(\BX) = \sum_{k=1}^T \BV_k^H \BX \BV_k, ~~  \widetilde{\cL}(\BY) = \sum_{k=1}^T \BV_k \BY \BV_k^H,
    \end{align}
    where $\BV_k \in \C^{N \times M}$ and $T \leq NM$.
\end{proposition}
\textit{Proof:} We will prove this using Lemma \ref{lem: Choi theorem}. The matrix $\BXi_{\widetilde{\cL}}$ is given by
\begin{align}
    \BXi_{\widetilde{\cL}} = \sum_{i, j} \BE_{ij} \otimes \widetilde{\cL}(\BE_{ij}) = \frac{1}{M} \sum_{i, j, a, b} \BE_{ij} \otimes (\BC_{ab} [\BE_{ij}]_{a, b}) = \frac{\BC}{M},
\end{align}
which is Hermitian non-negative. By Lemma \ref{lem: Choi theorem}, $\widetilde{\cL}$ therefore admits the representation $\widetilde{\cL}(\BY)= \sum_{k}\BV_k \BY \BV_k^H$. Define the map $\cL': \C^{N \times N} \mapsto \C^{M \times M}$ as  $\cL'(\BX) = \sum_k \BV_k^H \BX \BV_k$. Then, for any $\BX \in \C^{N \times N}$ and $\BY \in \C^{M \times M}$, we have
\begin{align}
    \innerprod{\cL'(\BX)}{\BY} = \sum_k \innerprod{\BV_k^H\BX\BV_k}{\BY} =  \innerprod{\BX}{\sum_k \BV_k \BY \BV_k^H} = \innerprod{\BX}{\widetilde{\cL}(\BY)}.
\end{align}
Since we already have $\innerprod{\cL\BX}{\BY} = \langle{\BX}, {\widetilde{\cL}\BY}\rangle$, it can be shown by contradiction that $\cL = \cL'$, which completes the proof. \qed 
\begin{proposition} 
\label{prop: bound on LI}
The map $\cL$ satisfies
     \begin{align}
         \norm{\cL(\BI_N)} \leq \frac{N}{M}\norm{\BC} , ~~ \lambda_{\min}(\cL(\BI_N)) \geq \frac{N}{M}\lambda_{\min}(\BC).
     \end{align}
\end{proposition}
\textit{Proof:} For the unit norm vector $\Bx \in \C^M$, we have 
\begin{align}
    \Bx^H \cL(\BI_N) \Bx = \frac{1}{M} \sum_{i, j} x_i^* x_j  \Tr (\BC_{ji})  = \frac{1}{M} \Tr (\Bx^\top \otimes \BI_N)\BC (\Bx^* \otimes \BI_N),
\end{align}
which implies 
$
    \Bx^H \cL(\BI_N) \Bx \leq \norm{\BC} \frac{1}{M} \Tr (\Bx^* \Bx^\top  \otimes \BI_N) = \frac{N}{M}\norm{\BC}
$ and $
    \Bx^H \cL(\BI_N) \Bx \geq \lambda_{\min}(\BC) \frac{1}{M} \Tr (\Bx^* \Bx^\top  \otimes \BI_N) = \frac{N}{M}\lambda_{\min}(\BC)
$. Therefore, we have completed the proof. \qed
\begin{proposition}
\label{prop: bound on the error LE term}
    For any  $z \in \C - \R_+$, we have the following bounds for the error term 
    \begin{align}
        \norm{{\cL}(\BCE(z))} = \mathcal{O}\left(\frac{P(z)\mathsf{E}_{N}}{N}\right), ~~~
         \norm{\Im({\cL}(\BCE(z)))}  = \mathcal{O}\left(\frac{\abs{\Im(z)}P(z)\mathsf{E}_{N}}{\dz N}\right). \label{eq: LE spectral norm bound}
    \end{align}
\end{proposition}
\textit{Proof:} We first evaluate the bilinear form $\Bx^H {\cL}(\BCE)\By$ for $\Bx, \By \in \C^{M}$. Expanding the map ${\cL}$ yields 
\begin{align}
    \Bx^H {\cL}(\BCE)\By = \frac{1}{M} \sum_{i, j} x_i^*  \Tr (\BC_{ji} \BCE) y_j = \frac{1}{M} \Tr \left\{ \left[ (\By^\top \otimes \BI_N) \BC (\Bx^* \otimes \BI_N) \right] (\E \BQ - \BT_\alpha) \right\} = \frac{1}{M}\Tr \BA_{C, x, y} (\E \BQ - \BT_\alpha),
\end{align}
where $\BA_{C, x, y} = (\By^\top \otimes \BI_N) \BC (\Bx^* \otimes \BI_N)$. Since $\lnorm{\BA_{C, x, y}} \leq  \norm{\By} \norm{\BC} \norm{\Bx}$, the first assertion in \eqref{eq: LE spectral norm bound} follows from \eqref{eq: trace bound on E} by taking the supremum over the unit vectors. The second assertion follows by the same argument as in \cite[Lemma B.1]{hachem2013bilinear}. \qed
\par
Define the difference matrix $\BDelta = \BT_\alpha - \BT$. Using the resolvent identity $\BX - \BY = \BY(\BY^{-1}- \BX^{-1})\BX$ yields
\begin{align}
    \BDelta & = \BT \widetilde{\cL}(z\widetilde{\BT}_\alpha - z\widetilde{\BT}) \BT_\alpha = z^2\BT \widetilde{\cL}(\widetilde{\BT} \cL( \BDelta + \BCE )\widetilde{\BT}_\alpha) \BT_\alpha. \label{eq: identity for Delta}
\end{align}
Applying $\cL$ to both sides of the above equation, we have
\begin{align}
    \cL(\BDelta) = \cF_{I, \alpha}[\cL(\BDelta)] + \cF_{I, \alpha}[\cL(\BCE)], \label{eq: ope system for Delta}
\end{align}
where the linear map $\cF_{I, \alpha}: \C^{M \times M} \mapsto \C^{M \times M}$ is defined as $\cF_{I, \alpha}(\cdot) = z^2 \cL(\BT\widetilde{\cL}(\widetilde{\BT} (\cdot)\widetilde{\BT}_\alpha)\BT_\alpha)$. To solve the matrix equation above, the invertibility of the map $\mathrm{Id}_M - \cF_{I, \alpha}$ must be established. Since $\cF_{I, \alpha}$ depends simultaneously on variables related to $\BT$ and $\BT_\alpha$,  its analysis is challenging.
To this end, we derive the auxiliary equations that involve only $\BT$ and $\BT_\alpha$ in the following. Applying the resolvent identity and $\Im(\BX) = (\BX - \BX^H) / 2\jmath$, for $z \in \C_+$, we have 
\begin{align}
    \cL(\Im(\BT)) = \Im(z) \cL\left(\BT\BT^H \right) + \mathcal{F}[\cL(\Im(\BT))], \label{eq: operator equation LT}
\end{align}
where $\mathcal{F}(\cdot) = |z|^2 \cL ( \BT \widetilde{\cL}(\widetilde{\BT} (\cdot) \widetilde{\BT}^H)\BT^H )$.  According to Propositions \ref{prop: sp norm bound on elements in Sn}, \ref{prop: bound on norm eigenvalue T}, \ref{prop: bound on LI}, and Lemma \ref{lem: bound on positive operator}, we have $\rho(\cF) < 1$ and 
\begin{align}
    \norm{(\mathrm{Id}_M - \cF)^{-1}}_{op} \leq \frac{\norm{\cL(\Im(\BT))}}{\Im(z) \lambda_{\min}(\cL(\BT\BT^H))} \leq \frac{\norm{\cL(\BI_N)} \norm{\Im(\BT)}}{\Im(z) \lambda_{\min}(\cL(\BI_N)) \lambda_{\min} \left( \BT\BT^H \right)} \leq \overline{K} \frac{(|z|^2 + \eta^2)^2}{\dz^4}, \label{eq: bound on z F}
\end{align}
where $\overline{K} > 0$ is a constant. We note that \eqref{eq: bound on z F} holds for $z < 0$. This follows by taking the limit $\lim_{\Im z \downarrow 0}$ in \eqref{eq: bound on z F} and using a continuation argument.
By a similar calculation as in \eqref{eq: operator equation LT}, we obtain 
\begin{align*}
    {\cL}(\Im(\E \BQ)) 
    &= \Im({\cL}(\BCE)) + \Im(z) {\cL}(\BT_\alpha^H \BT_\alpha) + \mathcal{F}_{\alpha} [\cL(\Im(\E \BQ))], \numberthis
\end{align*}
where the map $\mathcal{F}_{\alpha}(\cdot) = |z|^2 {\cL} ( \BT_\alpha^H  \widetilde{\cL}(\widetilde{\BT}_\alpha^H (\cdot) \widetilde{\BT}_\alpha) \BT_\alpha)$. According to Proposition \ref{prop: bound on the error LE term}, there exists a constant $\widetilde{K} > 0$ such that for $z$ in the region
\begin{align}
    \mathfrak{R} = \left\{z \in \C - \R_+ : \frac{ \widetilde{K}\mathsf{E}_N P(z)   (|z|^2 + \eta^2)^2}{N \dz^3} \leq 1 \right\},
\end{align}
there holds 
$\norm{\Im(\cL(\BCE))} \leq  \frac{K\lambda_{\min}(\cL(\BI_N)) \Im(z) \dz^2}{2(|z|^2 + \eta^2)^2}$ ($K$ is defined in Proposition \ref{prop: bound on norm eigenvalue T}). As a result, for $z \in \mathfrak{R}$, we have 
\begin{align*}
    \norm{(\mathrm{Id}_M - \cF_\alpha)^{-1}}_{op} &\leq \frac{\norm{\cL(\Im(\E \BQ))}}{\lambda_{\min}[\Im({\cL}(\BCE)) + \Im(z) {\cL}(\BT_\alpha^H \BT_\alpha)]} \\
    &\leq \frac{\norm{\cL(\BI_N)} \norm{\Im \E \BQ}}{\Im(z) \lambda_{\min}(\cL(\BI_N)) \lambda_{\min}(\BT_\alpha\BT_\alpha^H) - \norm{\Im(\cL(\BCE))}} \leq \frac{2\overline{K}(|z|^2 + \eta^2)^2}{\dz^4}. \numberthis \label{eq: bound on z Falp}
\end{align*}
With the representation in Proposition \ref{prop: representation of L cL}, we have 
\begin{align}
    &\cF(\cdot ) = \sum_{k, l} |z|^2 \BV_k^H \BT   \BV_l \widetilde{\BT} (\cdot) \widetilde{\BT}^H \BV_l^H \BT^H \BV_k = \sum_{ \boldsymbol{\mathsf{j}} = (k, l) } \BV_{\indvj}(\cdot) \BV_{\indvj}^H,  \\
    &\cF_{\alpha}(\cdot) = \sum_{k, l}  |z|^2 \BV_k^H \BT_\alpha^H   \BV_l \widetilde{\BT}_\alpha^H (\cdot) \widetilde{\BT}_\alpha \BV_l^H \BT_\alpha \BV_k = \sum_{\indvj = (k, l)} \BV^{\alpha}_{\indvj} (\cdot) (\BV_{\indvj}^\alpha)^H,
\end{align}
with $\BV_{(k, l)} = z \BV_k^H \BT   \BV_l \widetilde{\BT}$ and $ \BV_{(k, l)}^\alpha = z^*\BV_k^H \BT_\alpha^H   \BV_l \widetilde{\BT}_\alpha^H $. By the definition of $\cF_{I, \alpha}$, we have 
\begin{align}
    &\cF_{I, \alpha}(\cdot) = \sum_{k, l} z^2 \BV_k^H \BT  \BV_l \widetilde{\BT}(\cdot) \widetilde{\BT}_\alpha \BV_l^H \BT_\alpha \BV_k = \sum_{\indvj = (k, l)} \BV_{\indvj} (\cdot) (\BV_{\indvj}^\alpha)^H.
\end{align}
For a multi-index $(\indvj_1, \indvj_2, \ldots, \indvj_k)$, define $\BA_{(\indvj_1, \indvj_2, \ldots, \indvj_k)} = \prod_{i = 1}^k \BA_{\indvj_i}$. For any $\BX \in \C^{M \times M}$ and any positive integer $k$, we have 
\begin{align*}
   \abs{ \Bx^H \cF_{I, \alpha}^k(\BX) \By} &= \abs{ \sum_{\boldsymbol{\mathsf{J}}_k = (\indvj_1, \indvj_2, \ldots, \indvj_k)} \Bx^H\BV_{\boldsymbol{\mathsf{J}}_k} \BX \left( \BV_{\boldsymbol{\mathsf{J}}_k}^\alpha \right)^H \By} \leq \norm{\BX}\sum_{\boldsymbol{\mathsf{J}}_k} \norm{\Bx^H\BV_{\boldsymbol{\mathsf{J}}_k}} \norm{\By^H\BV_{\boldsymbol{\mathsf{J}}_k}^\alpha}\\
    & \leq \norm{\BX} \left(\sum_{\boldsymbol{\mathsf{J}}_k}\Bx^H \BV_{\boldsymbol{\mathsf{J}}_k}\BV_{\boldsymbol{\mathsf{J}}_k}^H \Bx \right)^{\frac{1}{2}}   \left( \sum_{\boldsymbol{\mathsf{J}}_k}\By^H \BV_{\boldsymbol{\mathsf{J}}_k}^\alpha (\BV_{\boldsymbol{\mathsf{J}}_k}^\alpha)^H \By \right)^{\frac{1}{2}} = \norm{\BX} \left( \Bx^H \cF^k(\BI_M) \Bx\right)^{\frac{1}{2}}\left( \By^H \cF^k_\alpha(\BI_M) \By \right)^{\frac{1}{2}}. \numberthis \label{eq: bound on the operator FIalp}
\end{align*}
By taking the supremum over unit vectors $\Bx$, $\By$ and unit spectral norm matrix $\BX$ on both sides of the above equation, we obtain 
\begin{align}
    \norm{\cF_{I, \alpha}^k}_{op}^{\frac{1}{k}} \leq \sqrt{\norm{\cF^k(\BI_M)}^{\frac{1}{k}} \norm{\cF^k_\alpha(\BI_M)}^{\frac{1}{k}} } \overset{(a)}{=} \sqrt{\norm{\cF^k}_{op}^{\frac{1}{k}}  \norm{\cF_\alpha^k}_{op}^{\frac{1}{k}} },
\end{align}
where step $(a)$ follows from the Russo–Dye theorem \cite{russo1966note}. By Gelfand's formula, we have, for $z \in \mathfrak{R}$, 
\begin{align}
\rho(\cF_{I, \alpha}) = \lim_{k \to \infty} \norm{\cF_{I, \alpha}^k}_{op}^{\frac{1}{k}} \leq \sqrt{\rho(\cF_\alpha) \rho(\cF)} < 1,
\end{align}
which implies $\mathrm{Id}_M - \cF_{I, \alpha}$ is invertible. According to the Neumann series, we have 
$
    \left( \mathrm{Id}_M - \cF_{I, \alpha} \right)^{-1} = \sum_{j = 0}^{\infty} \cF_{I, \alpha}^j
$. Hence, by \eqref{eq: bound on z F}, \eqref{eq: bound on z Falp}, and \eqref{eq: bound on the operator FIalp}, we obtain, for $z \in \mathfrak{R}$,
\begin{align*}
    &\norm{\left( \mathrm{Id}_M - \cF_{I, \alpha} \right)^{-1}}_{op} = \sup_{\norm{\BX} = 1}\norm{\left( \mathrm{Id}_M - \cF_{I, \alpha} \right)^{-1}(\BX)} = \sup_{\norm{\BX} = 1} \sup_{\norm{\Bx} = \norm{\By} = 1} \abs{\sum_{j=0}^\infty \Bx^H \cF_{I, \alpha}^j(\BX) \By} \\
    &\leq \sup_{\norm{\Bx} = \norm{\By} = 1} \left(\Bx^H \sum_{j=0}^\infty \cF^j(\BI_M) \Bx \right)^{\frac{1}{2}}\left(\By^H \sum_{j=0}^\infty \cF^j_\alpha(\BI_M) \By \right)^{\frac{1}{2}} \leq \sqrt{ \norm{\left( \mathrm{Id}_M - \cF\right)^{-1}}_{op} \norm{\left( \mathrm{Id}_M - \cF_\alpha\right)^{-1}}_{op} } \\
    & \leq \frac{\sqrt{2}\overline{K}(|z|^2 + \eta^2)^2}{\dz^4}, \numberthis
\end{align*}
By \eqref{eq: ope system for Delta}, we can get $\cL(\BDelta) = (\mathrm{Id}_M - \cF_{I, \alpha})^{-1} \cF_{I, \alpha} [\cL(\BCE)] = -\cL(\BCE) + (\mathrm{Id}_M - \cF_{I, \alpha})^{-1}[\cL(\BCE)]$ and 
\begin{align}
\norm{\cL(\BDelta)} \leq \norm{\cL(\BCE)} + \norm{\left(\mathrm{Id}_M - \cF_{I, \alpha} \right)^{-1}}_{op} \norm{\cL(\BCE)} \leq  \frac{K_1(|z|^2 + \eta^2)^2 P(z) \mathsf{E}_N}{\dz^4 N}, ~~~ z \in \mathfrak{R},
\end{align}
where $K_1 > 1$ is a constant and we use $|z|^2 + \eta^2 \geq \dz^2$. By the identity \eqref{eq: identity for Delta}, we can get 
\begin{align}
    \abs{ \Tr \BA \left( \BT_\alpha - \BT \right) } \leq  N|z|^2 \norm{\BA} \norm{\BT}\norm{\BT_\alpha}\norm{\widetilde{\BT}}\norm{\widetilde{\BT}_\alpha} \norm{\widetilde{\cL}}_{op} \left( \norm{\cL(\BDelta)} + \norm{\cL(\BCE)} \right) \leq  \frac{K_2 (|z|^2 + \eta^2)^3P(z)\mathsf{E}_N}{\dz^8}, ~~ z \in \mathfrak{R}.
\end{align}
On the other hand, for $z \in \C - \R_+ - \mathfrak{R}$, we have $1 \leq  \widetilde{K} \mathsf{E}_N P(z)(|z|^2 + \eta^2)^2 / (N\dz^3)$
\begin{align}
    \abs{\Tr \BA \BDelta} \leq \frac{2N \norm{\BA}}{\dz} \times 1 \leq \frac{2\norm{\BA}\widetilde{K}  (|z|^2 + \eta^2)^2P(z)\mathsf{E}_N }{\dz^4}.
\end{align}
Using $\max(1, (|z|^2 + \eta^2)^2 / \dz^4, (|z|^2 + \eta^2)^3 / \dz^8) \leq (\eta^{-2} + 2)(|z|^2 + \eta^2)^4/ \dz^8$, we obtain, for any $z \in \mathbb{C} - \R_+$, 
\begin{align}
    \abs{\Tr \BA (\E \BQ - \BT)} \leq \abs{\Tr \BA \BCE} + \abs{\Tr \BA \BDelta} = \cO\left(\frac{(|z|^2 + \eta^2)^4 P(z)\mathsf{E}_N}{\dz^8}\right).
\end{align}
By expanding the polynomials in the RHS of the above inequality, we have 
\begin{align*}
    \frac{(|z|^2 + \eta^2)^4 P(z)}{ \dz^8} \leq \frac{1}{\dz^4}  \frac{(|z|^2 + \eta^2)^4}{ \dz^8}\left( \frac{|z|}{\dz} + \frac{|z| + |z|^2}{\dz^2} + \frac{|z|^2}{\dz^3} + \frac{|z|^2 + |z|^3}{\dz^4} \right) \leq K_\eta \frac{\Psi(z)}{\dz^4}, \numberthis
\end{align*}
for some constant $K_\eta$. Therefore, 
Theorem \ref{thm: de for the resolvent} is proved. \qed

\section{Proof of Lemma \ref{lem: bound on banded Stieltjes transform}}
\label{app: lemma bound on banded Stieltjes transform}
In this section, we use the perturbation method to bound the absolute values of the entries of $\BG$. Intuitively, since $\BL$ has a finite bandwidth, the effect of a small perturbation on $\BL$ is limited. Hence, $\BG$ can be effectively constrained via its perturbed version. 
\par
Recall the notation $\BCI$ and $\cB_I$ in Lemma \ref{lemm: band berstein bound}. Fix $z \in \C - \R_+$ and let $\alpha \in \R$. Define the diagonal weight matrix $\BCD_\alpha = e^{\alpha \BCI} =  \diag(1, e^{\alpha}, e^{2\alpha},  \ldots, e^{(N-1)\alpha})$. For a matrix $\BX \in \C^{N \times N}$, denote its perturbed version as $\BX_\alpha = \BCD_\alpha \BX \BCD_{-\alpha}$.  Then, for any unit vectors $\Bx, \By \in \C^N$, we have
\begin{align}
    \Bx^H(\BL_\alpha - \BL) \By = \sum_{i, j} x_i^* y_j (e^{\alpha(i - j)} - 1) [\BL]_{i, j}  = \sum_{i, j} \sum_{k = 1}^{\infty} \frac{\alpha^k (i - j)^k}{k!} x_i^* y_j [\BL]_{i, j} = \sum_{k = 1}^{\infty} \frac{\alpha^k}{k!} \Bx^H \cB_I^k(\BL) \By.
\end{align}
Using Lemma \ref{lemm: band berstein bound}, we obtain $\norm{\cB_I^k(\BL)} \leq W^k \norm{\BL}$ and 
\begin{align}
\norm{\BL_\alpha - \BL} \leq \sum_{k=1}^{\infty} \frac{|\alpha|^k W^k}{k!} \norm{\BL} = (e^{|\alpha|W} - 1) \norm{\BL}.
\end{align}
It can be verified that when $|\alpha|W \leq \frac{1}{2}$, there holds $|e^{\alpha W} - 1| \leq 2 |\alpha|W$. Since $\BL(z) \in \mathcal{S}_N(\R_+)$, we have $\norm{\BL} \leq  \norm{\Bmu_L(\R_+)} \dz^{-1}$. Then, we can get 
\begin{align}
    \norm{\BL_\alpha - \BL} \leq \frac{2 |\alpha| W\norm{\Bmu_L(\R_+)}}{\dz}.
\end{align}
By linear algebra, we have 
\begin{align}
    \BG_{\alpha}^{-1} = \BG^{-1} + (\BG_{\alpha}^{-1} - \BG^{-1}) = \BG^{-1} - z(\BL_{\alpha} - \BL) = \BG^{-1}(\BI_N - z\BG(\BL_{\alpha} - \BL)).
\end{align}
Since $\BL(z) \in \mathcal{S}_N(\R_+)$, Propositions \ref{prop: L in Sn} and \ref{prop: sp norm bound on elements in Sn} imply that $\BG(z) \in \mathcal{S}_N (\mathbb{R}_+)$ and $\norm{\BG(z)} \leq \dz^{-1}$. Consequently, if $\alpha$ satisfies $|\alpha|W \leq  \mu_z = \min( \frac{\dz^2}{4|z|\norm{\Bmu_L(\R_+)}}, \frac{1}{2}) $, then
$
    \norm{z \BG (\BL_{\alpha} - \BL)} \leq |z| \norm{\BG} \norm{\BL_\alpha - \BL} \leq \frac{1}{2}
$,
which implies the spectral radius $\rho(z \BG (\BL_{\alpha} - \BL)) \leq  \frac{1}{2}$. Thus, the matrix $\BI_N - z\BG(\BL_{\alpha} - \BL)$ is invertible.
Expanding the inverse by the Neumann series gives
\begin{align}
    \BG_\alpha = \sum_{j=0}^{\infty} \left( z \BG (\BL_\alpha - \BL) \right)^j \BG,
\end{align}
and 
\begin{align}
    \norm{\BG_\alpha} \leq  \sum_{j=0}^{\infty} \norm{z \BG (\BL_\alpha - \BL)}^j \norm{\BG} \leq \frac{2}{\dz}.
\end{align}
Hence, the relation $[\BG]_{i, j} = e^{(j-i)\alpha} [\BG_\alpha]_{i, j}$ holds for all $|\alpha|W \leq \mu_z$. This yields 
\begin{equation}
    |[\BG]_{i, j}| \leq \inf_{|\alpha|W \leq \mu_z} e^{(j-i)\alpha}  \norm{\BG_\alpha} \leq \frac{2}{\dz}\exp\left( -  \frac{\mu_z\abs{i - j}}{W + 1}\right) .
\end{equation}
Therefore, we have completed the proof of Lemma \ref{lem: bound on banded Stieltjes transform}. \qed

\section{Proof of Lemma \ref{lem: variance control}}
\label{app: proof of lemma variance control}
In this section, we will use the Poincaré–Nash inequality \cite{Hachem2008ANewApproach} to prove Lemma \ref{lem: variance control}. Specifically, we first bound the variance of the linear functionals of the resolvent, and then extend the variance bound to central moments of arbitrary integer order by induction.  Since \eqref{eq: moment bound on trace AQ} and \eqref{eq: moment bound on trace AQHbHH} can be proved analogously, in this section, we only give the proof of \eqref{eq: moment bound on trace AQ} and omit that of \eqref{eq: moment bound on trace AQHbHH} for simplicity.
\subsection{Variance Case}
By the Poincaré-Nash inequality, we have 
\begin{align*}
   & \Var (\Tr\BA\BQ ) \leq \E \left[ \sum_{i, j, a, b}  \frac{\partial \Tr \BA \BQ}{\partial H_{i,j}} \left(\E H_{i, j}H_{a, b}^* \right) \left( \frac{\partial \Tr \BA \BQ}{\partial H_{a, b}} \right)^* + \left( \frac{\partial \Tr \BA \BQ}{\partial H_{i,j}^*} \right)^* \left(\E H_{i, j}H_{a, b}^* \right) \frac{\partial \Tr \BA \BQ}{\partial H_{a, b}^*} \right] \\
    &= \E \left[ \sum_{i, j, a, b} [\BH^H\BQ\BA\BQ]_{j, i} \frac{[\BC_{jb}]_{i, a}}{M} [\BQ^H\BA^H\BQ^H\BH]_{a, b} \right] + \E \left[ \sum_{i, j, a, b} [\BH^H\BQ^H\BA^H\BQ^H]_{j, i} \frac{[\BC_{jb}]_{i, a}}{M} [\BQ\BA\BQ\BH]_{a, b} \right] \\
    &= \sum_{j, b} \frac{1}{M}\E \left[ \left[\BH^H\BQ\BA\BQ\BC_{jb} \BQ^H\BA^H\BQ^H\BH\right]_{j, b} +  \left[\BH^H\BQ^H\BA^H\BQ^H\BC_{jb}\BQ\BA\BQ\BH \right]_{j, b}   \right]. \numberthis\label{eq: var trace AQ bound1}
\end{align*}
Denote two vectors $\Bx = \Vecm(\BQ^H\BA^H\BQ^H\BH)$ and $\overline{\Bx} = \Vecm(\BQ\BA\BQ\BH)$.  Then, the RHS of  \eqref{eq: var trace AQ bound1} can be written as $\frac{1}{M} \E ( \Bx^H \BC\Bx + \overline{\Bx}^H \BC \overline{\Bx} )$. Thus, we have
\begin{align*}
    \Var(\Tr \BA \BQ) &\leq \frac{1}{M}\E\left[ \Bx^H \BC \Bx + \overline{\Bx}^H \BC\overline{\Bx}\right] \leq \frac{\norm{\BC}}{M} \E \left[ \norm{\Bx}^2 + \norm{\overline{\Bx}}^2\right] \\
    &\leq \frac{\norm{\BC}}{M} \E\left[ \Tr \BH^H\BQ\BA\BQ\BQ^H\BA^H\BQ^H\BH + \Tr \BH^H \BQ^H\BA^H\BQ^H\BQ\BA\BQ\BH \right] \\
    &\overset{(a)}{\leq} \frac{2 \Tr \BA\BA^H \norm{\BC}}{M} \E \left[ \norm{\BQ}^4 \norm{\BH}^2 \right] \overset{(b)}{=} \mathcal{O}\left(\frac{\Tr \BA\BA^H}{M \mathsf{d}_z^4}\right), \numberthis \label{eq: variance bound for var TrAQ}
\end{align*}
where step $(a)$ follows from the inequality $|\Tr \BX \BY| \leq \norm{\BX} \Tr \BY$ for Hermitian non-negative matrix $\BY$. Step $(b)$ is due to the fact that $\norm{\BQ(z)}\leq \mathsf{d}_z^{-1}$ and Theorem \ref{thm: bound norm}. Since $M$ and $N$ are of the same order, \eqref{eq: moment bound on trace AQ} holds for $p = 2$.
\subsection{General Central Moment Case}
To simplify notation, we define $X = \Tr \BA \BQ$, $\mathring{X} = X - \E X$, and $\mathsf{a} = N^{\frac{1}{2}}{\mathsf{d}_z^2} / {\sqrt{\Tr \BA\BA^H}}$. By the Cauchy–Schwarz inequality, we have $\E |\mathsf{a}\mathring{X}| \leq \E^{1/2}|\mathsf{a} \mathring{X}|^2 = \mathcal{O}(1)$.
Starting from \(p=1,2\), we proceed by induction on \(k\ge2\). Assume that the bound \eqref{eq: moment bound on trace AQ} holds for all $p=1, 2, \ldots, 2k-2$. We now establish the case $p = 2k$.  Observing that $\Var(x^k) = \E|x|^{2k} - |\E x^k|^2$, we can evaluate $\E |x|^{2k}$ by bounding $\Var(x^k)$ and $|\E x^k|^2$ separately.  Using the Poincaré-Nash inequality, we obtain
\begin{align*}
    \Var(\mathring{X}^k) &\leq \E \sum_{i, j, a, b} \frac{\partial \mathring{X}^k}{\partial H_{i,j}} \frac{[\BC_{jb}]_{i,a}}{M} \left( \frac{\partial \mathring{X}^k}{\partial H_{a, b}} \right)^* + \left( \frac{\partial \mathring{X}^k}{\partial H_{i,j}^*} \right)^* \frac{[\BC_{jb}]_{i,a}}{M} \frac{\partial \mathring{X}^k}{\partial H_{a, b}^*} \\
    &= \frac{k^2}{M} \E \left[ \abs{\mathring{X}}^{2k-2} \left(  \sum_{j,b}\left[\BH^H\BQ\BA\BQ\BC_{jb} \BQ^H\BA^H\BQ^H\BH\right]_{j, b} +  \left[\BH^H\BQ^H\BA^H\BQ^H\BC_{jb}\BQ\BA\BQ\BH \right]_{j, b}  \right) \right] \\
    & \overset{(a)}{\leq} \frac{k^2K  }{\mathsf{a}^2} \E \left[  \abs{\mathring{X}}^{2k-2} \norm{\BH}^2 \right], \numberthis
\end{align*}
where step $(a)$ follows the same argument as in \eqref{eq: variance bound for var TrAQ} and $K$ is a constant independent of $N$ and $M$. Using the induction hypothesis $\E |\mathsf{a}\mathring{X}|^{2k-2} = \mathcal{O}_k(1)$, the tail bound in \eqref{eq: proba tail bound for norm H}, and Hölder's inequality $\E|x|^{2k-2}|y||z| \leq \E^{1-\frac{1}{k}} |x|^{2k} \E^{\frac{1}{2k}} |y|^{2k}\E^{\frac{1}{2k}} |z|^{2k}$, the following is obtained for sufficiently large $K'_1$
\begin{align*}
    \mathsf{a}^{2k}\Var(\mathring{X}^k)  &\leq {k^2K(K_1')^2}  \E \left[  \abs{\mathsf{a}\mathring{X}}^{2k-2}\right] + k^2K\E \left[ \abs{\mathsf{a}\mathring{X}}^{2k-2} \norm{\BH}^2 \Ind\{ \norm{\BH} > K'_1\}\right] \\
    & {\leq} \mathcal{O}_k(1) \left[ 1 +  \left(\E  \abs{\mathsf{a}\mathring{X}}^{2k} \right)^{1 - \frac{1}{k}} \left(\E \norm{\BH}^{4k} \right)^{\frac{1}{2k}}\left(\PP\{ \norm{\BH} > K'_1\} \right)^{\frac{1}{2k}} \right] \\
    & \leq \mathcal{O}_k(1)\left[1 +  e^{-\frac{K_2'N}{2k}}\left(\E  \abs{\mathsf{a}\mathring{X}}^{2k} \right)^{1 - \frac{1}{k}} \right], \numberthis
\end{align*}
where $K'_2 > 0$ is a constant. Then, by the identity $\E|\mathring{X}|^{2k} =  \Var(\mathring{X}^k) + |\E \mathring{X}^k|^2$ and the induction hypothesis $\E |\mathsf{a}\mathring{X}|^{k} = \cO_k(1)$, the following is obtained
\begin{align*}
       \left( \E  \abs{\mathsf{a}\mathring{X}}^{2k} \right)^{\frac{1}{k}} \leq \mathcal{O}_k(1)\left[\left(\E  \abs{\mathsf{a}\mathring{X}}^{2k} \right)^{-1 + \frac{1}{k}} +  e^{-\frac{K_2'N}{2k}} \right]. \numberthis \label{eq: bound on E aX_2k}
\end{align*}
Next, we show by contradiction that $\E |\mathsf{a}\mathring{X}|^{2k}$ is uniformly bounded by a constant, which only depends on $k$. Suppose the contrary, i.e.,  $\limsup_{N \to \infty}\E |\mathsf{a}\mathring{X}|^{2k} = \infty$. Take the subsequence $\{N_m\}$ such that $ \E |\mathsf{a} \mathring{X}|^{2k} \to \infty$ as $m \to \infty$. Then, the LHS of \eqref{eq: bound on E aX_2k} becomes infinite, while the RHS tends to 0. This yields a contradiction, thus establishing $\E |\mathsf{a}\mathring{X}|^{2k} = \mathcal{O}_k(1)$. By Hölder's inequality $\E^\frac{1}{a}|x|^a \geq \E^{\frac{1}{b}}|x|^b$ for $a > b > 0$, we have $\E|\mathsf{a}\mathring{X}|^{2k-1} \leq \E^{\frac{2k - 1}{2k}}|\mathsf{a} \mathring{X}|^{2k}  = \mathcal{O}_k(1)$. This completes the inductive argument for $p=2k-1$ and $p = 2k$. Therefore,  we have completed the proof of \eqref{eq: moment bound on trace AQ}. \qed

\section{Proof of Proposition \ref{prop: bound on norm eigenvalue T}}
\label{app: proof of prop bound on norm eigenvalue T}
We will only discuss the case for $\BT$ for brevity. To prove Proposition \ref{prop: bound on norm eigenvalue T}, the following identity is useful. Let $\BF \in \cS_n(\R_+)$ with associated matrix-valued measure $\Bmu$. Then
\begin{align}
    \int_{\R_+} \lambda \Bmu(\dd \lambda) = \lim_{y \to \infty} \Re \left\{ -\jmath y \left[ \jmath y \BF(\jmath y) + \Bmu(\R_+) \right]  \right\}.
\end{align}
This identity is a natural extension of the scalar measure case in \cite{AblaNoEigen} and follows from the monotone convergence theorem. Denote the underlying measure of $\BT$ as $\Bnu_N$ (we reintroduce the index $N$ to emphasize the sequence of measures $\{\Bnu_N\}_{N \geq 1}$). By \eqref{eq: DE systems}, we have 
\begin{align}
    -\jmath y (\jmath y \BT(\jmath y) + \BI_N) = \widetilde{\cL}(-\jmath y \widetilde{\BT}(\jmath y))(\BI_N + \widetilde{\cL}(\widetilde{\BT}(\jmath y)))^{-1}.
\end{align} 
As shown in the proof of Proposition \ref{prop: exist unique DE},  $\lim_{y \to \infty} -\jmath y \widetilde{\BT}(\jmath y) = \BI_M$ and consequently $\int_{\R_+} \lambda \Bnu_N(\dd \lambda) = \widetilde{\cL}(\BI_M)$. By Assumption \ref{assumpt: A2} and the inequality $\max_j \lnorm{\BC_{jj}} \leq \lnorm{\BC}$, we have 
\begin{align}
     \widetilde{\cL}(\BI_M)  = \frac{1}{M} \sum_{j=1}^M \BC_{jj} \preceq \norm{\BC} \BI_N.
\end{align}
Hence, the sequence of matrix-valued measures $\{\Bnu_N\}$ is tight. As a result, there exists a constant $\eta > 0$ such that $\lambda_{\min}(\Bnu_N([0, \eta])) \geq \frac{1}{2}$ for all $N$. In fact, this can be shown by the Markov inequality
\begin{align}
    \BI_N - \Bnu_N([0, \eta]) =  \int_{(\eta, \infty)} \Bnu_N(\dd \lambda) \preceq \int_{(\eta, \infty)} \frac{\lambda}{ \eta} \Bnu_N(\dd \lambda) \preceq \frac{\widetilde{\cL}(\BI_M)}{\eta} \preceq \frac{\norm{\BC}}{\eta} \BI_N.
\end{align}
For a unit vector $\Bx \in \C^N$, we denote the scalar measure $\nu_{N}^x = \Bx^H \Bnu_N \Bx$. Then, we have $\nu_{N}^x([0, \eta]) = \Bx^H \Bnu_N([0, \eta]) \Bx \geq \frac{1}{2}$.  For $z \in \C - \R_+$, we can write
\begin{align}
    \Bx^H \BT(z) \Bx = \int_{\R_+} \frac{(\lambda - \Re z)}{|\lambda - z|^2} \nu^x_N(\dd \lambda) + \jmath \int_{\R_+} \frac{\Im z}{|\lambda - z|^2} \nu^x_N(\dd \lambda) := a + \jmath b.
\end{align}
We then derive a lower bound for $|\Bx^H \BT(z) \Bx|$ by considering two cases.
\par
\textit{Case I: $z \in \C - \R_+$ and $\Re z \geq 0$.} In this case, we have $|\Im z| = \dz$ and 
\begin{align}
\label{eq: lower bound xHTx}
    \abs{\Bx^H \BT(z) \Bx} \geq |b| = \dz \int_{\R_+} \frac{\nu^x_N(\dd \lambda)}{(\lambda - \Re z)^2 + (\Im z)^2} \overset{(a)}{\geq}  \dz \frac{\int_{[0, \eta]}  \nu^x_N(\dd \lambda)}{2(\eta^2 + (\Re z)^2) + (\Im z)^2} \geq \frac{\dz}{4(|z|^2 + \eta^2)},
\end{align}
where step $(a)$ follows from $(\lambda - \Re z)^2 \leq 2(\lambda^2 + (\Re z)^2) \leq 2(\eta^2 + (\Re z)^2)$ for $\lambda \in [0, \eta]$. 
\par
\textit{Case II: $z \in \C - \R_+$ and $\Re z < 0$.} In this case, we have $|\Im z| + |\Re z| \geq  \dz$ and
\begin{align}
    \abs{\Bx^H \BT(z) \Bx} \geq \frac{|a| + |b|}{\sqrt{2}} {\geq} \frac{|\Re z|  + |\Im z|}{\sqrt{2}}\int_{\R_+} \frac{\nu^x_N(\dd \lambda)}{|\lambda - z|^2} \overset{(a)}{\geq} \frac{\dz}{4\sqrt{2}(|z|^2 + \eta^2)},
\end{align}
where step $(a)$ uses the same denominator estimate as in \textit{Case I}, together with $\nu_N^x([0,\eta])\ge1/2$.
\par
Combining the two cases and using the inequality $
    \abs{\Bx^H \BT \Bx}^2 = |\Bx^H \BT^H \Bx \Bx^H \BT \Bx | \leq \norm{\Bx}^2 \Bx^H \BT^H \BT \Bx $ yields
\begin{align}
    \lambda_{\min}(\BT^H\BT) = \inf_{\norm{\Bx} = 1} \Bx^H \BT^H \BT \Bx \geq \inf_{\norm{\Bx} = 1}\abs{\Bx^H \BT \Bx}^2 \geq \frac{\dz^2}{32(|z|^2 + \eta^2)^2}.
\end{align}
Therefore, we have completed the proof of Proposition \ref{prop: bound on norm eigenvalue T}. \qed

\section{Proof of Theorem \ref{thm: EMI approximation}}
\label{app: proof of thm EMI approximation}
By taking the derivative and using the formula $\partial \log\det(\BX) = \langle\partial \BX , \BX^{-H}\rangle$, we obtain
\begin{align}
\cI(\sigma^2) = \int_{\sigma^2}^{\infty} -\frac{\partial \cI(x)}{\partial x} \dd x = \int_{\sigma^2}^{\infty} \left( \frac{N}{x} - \Tr \BQ(-x) \right) \dd x, \label{eq: Cauchy trlation}
\end{align}
which establishes a connection between the resolvent and the MI. We then show that $\overline{\cI}$ is a good approximation of $\E[\cI]$. Differentiating the deterministic approximation gives
\begin{align*}
    \frac{\partial \overline{\cI}(x)}{\partial x} &= \frac{\partial}{\partial x}\left[ \log\det\left(\BI_N + \widetilde{\cL}(\widetilde{\BT}(-x))\right) + \log\det\left(\BI_M + \cL(\BT(-x))\right) - x \innerprod{\cL(\BT(-x))}{\widetilde{\BT}(-x)} \right] \\
    & \overset{(a)}{=} \innerprod{\frac{\partial \widetilde{\cL}(\widetilde{\BT}(-x))}{\partial x}}{\left(\BI_N + \widetilde{\cL}(\widetilde{\BT}(-x))\right)^{-1}}  + \innerprod{\frac{\partial \cL(\BT(-x))}{\partial x}}{\left(\BI_M + \cL(\BT(-x))\right)^{-1}}  \\
    &- \innerprod{\cL(\BT(-x))}{\widetilde{\BT}(-x)} - \innerprod{\frac{\partial\cL(\BT(-x))}{\partial x}}{x\widetilde{\BT}(-x)} - x\innerprod{\cL(\BT(-x))}{\frac{\partial \widetilde{\BT}(-x)}{\partial x}} \\
    & \overset{(b)}{=} - \innerprod{\cL(\BT(-x))}{\widetilde{\BT}(-x)} \overset{(c)}{=}  \Tr \BT(-x)- \frac{N}{x}. \numberthis
\end{align*}
Here, step $(a)$ follows by $(\BT(-x))^H = \BT(-x)$ and $(\widetilde{\BT}(-x))^H = \widetilde{\BT}(-x)$. Step $(b)$ follows by the identities  $\langle \cL\BX, \BY \rangle = \langle \BX, \widetilde{\cL}\BY \rangle$ and $\widetilde{\cL}(\partial \BX) = \partial \widetilde{\cL}(\BX)$. Step $(c)$ is due to the resolvent identity $(x\BI_N)^{-1} - \BT(-x) = \BT(-x) \widetilde{\cL}[\widetilde{\BT}(-x)]$. Hence, by $\lim_{x \to +\infty} \overline{\cI}(x) = 0$, we obtain
\begin{align}
    \overline{\cI}(\sigma^2) = \int_{\sigma^2}^{\infty} \left( \frac{N}{x} - \Tr \BT(-x) \right) \dd x.
\end{align}

By setting $\BA = \BI_N$ in Theorem \ref{thm: de for the resolvent}, we have 
$
    \Tr \E\BQ(-x) = \Tr \BT(-x) + \mathcal{O}(1) \frac{\Psi(-x) \mathsf{E}_N}{x^4}
$.
Combining the above analysis and Fubini's theorem, we have 
\begin{align*}
    \E \cI(\sigma^2) - \overline{\cI}(\sigma^2) = \int_{\sigma^2}^{\infty}\Tr[\BT(-x) - \E \BQ(-x)] \dd x  = \cO(\mathsf{E}_N) \int_{\sigma^2}^{\infty} \frac{\Psi(-x)}{x^4} \dd x = \cO\left(\mathsf{E}_N \sum_{j=3}^{15} \frac{1}{\sigma^{2j}}\right). \numberthis
\end{align*}
Therefore, we have completed the proof of Theorem \ref{thm: EMI approximation}. \qed

\bibliographystyle{IEEEtran}
\bibliography{reference}
\end{document}